\documentclass[11pt,fleqn,a4paper]{article}

\usepackage{geometry}
\usepackage{newtxtext}
\usepackage{newtxmath}
\usepackage{graphicx}
\usepackage{wrapfig}

\usepackage[fleqn]{mathtools}
\usepackage{xspace}
\usepackage{gensymb} 
\usepackage{hyperref}
\usepackage{color}
\hypersetup{
 colorlinks=true,
 citecolor=blue
}

\usepackage{pifont}
\newcommand{\circled}[1]{\text{\ding{\the\numexpr 171 + #1}}}
\newcommand{\Checkmark}{\text{\ding{51}}\xspace}
\newcommand{\Cross}{\text{\ding{55}}\xspace}

\usepackage{dsfont} 

\newcommand{\modefun}{\Delta}

\newcommand{\pbox}[2][c]{\begin{tabular}[#1]{@{}l@{}}#2\end{tabular}}

\newcommand{\infrule}[3][\void]{%
  {\renewcommand\arraystretch{1.25}
    \ifx\void#1\else\raisebox{-7.8pt}{\IR{#1}}\hspace{0.5em}\fi
    \begin{array}[t]{@{\hspace*{1em}}c@{\hspace*{1em}}}#2\\\hline
      #3
    \end{array}    
}}
\newcommand{\infrulecond}[1]{\raisebox{-7.8pt}{\ $#1$}}

\newcommand{\IR}[1]{\text{\textsf{#1}}\xspace}

\definecolor{DodgerUniformBlue}{rgb}{0.0,0.353,0.612}
\newcommand{\define}[1]{\emph{\textcolor{DodgerUniformBlue}{#1}}}

\newcommand{\gammaone}{\gamma_{\mathrm{one}}}
\newcommand{\gammaleft}{\gamma_{\mathrm{left}}}
\newcommand{\gammaright}{\gamma_{\mathrm{right}}}

\newcommand{\compl}[1]{\mathrel{\overline{#1}}}
\newcommand{\PQ}[1]{{\boldsymbol{\mathsf{#1}}}}
\newcommand{\TO}[1]{\boldsymbol{\mathsf #1}}

\newcommand{\calT}{{\cal T}}
\newcommand{\calM}{{\cal M}}
\newcommand{\calF}{{\cal F}}

\newcommand{\opFalse}{\mathit{false}}
\newcommand{\opTrue}{\mathit{true}}

\def\P_#1{\PQ{P}_{#1}\,}

\newcommand{\X}{\TO{X}\,}
\newcommand{\U}{\,\TO{U}\,}
\newcommand{\R}{\,\TO{R}\,}
\newcommand{\W}{\,\TO{W}\,}
\newcommand{\F}{\TO{F}\,}
\newcommand{\G}{\TO{G}\,}

\renewcommand{\Pr}{\operatorname{Pr}}
\newcommand{\Prescribed}{\mathop{\mathit{Prescribed}}}
\newcommand{\Subtree}{\mathop{\mathit{Subtree}}}
\newcommand{\Runs}{\mathop{\mathit{Runs}}}
\newcommand{\Paths}{\mathop{\mathit{Paths}}}

\newcommand{\Succ}{\mathop{\mathit{Succ}}}
\newcommand{\Dist}{\mathop{\mathit{Dist}}}
\newcommand{\start}{\mathrm{start}}
\newcommand{\act}{\mathrm{act}}
\newcommand{\init}{\mathrm{init}}
\newcommand{\last}{\operatorname{last}}
\newcommand{\first}{\operatorname{first}}
\newcommand{\dom}{\operatorname{dom}}
\newcommand{\final}{\mathrm{final}}
\newcommand{\fin}{\mathrm{fin}}
\newcommand{\TABLEAU}{\mathop{\text{\textnormal{\textsc{Tableau}}}}}
\newcommand{\CHOOSE}{\mathop{\text{\textnormal{\textsc{Choose}}}}}
\newcommand{\GAMMA}{\mathop{\text{\textnormal{\textsc{Gamma}}}}}
\newcommand{\FORCE}{\mathop{\text{\textnormal{\textsc{Force}}}}}
\newcommand{\BSCC}{\mathop{\text{\textnormal{\textsc{Bscc}}}}}

\newcommand{\pair}[2]{\langle#1,#2\rangle}
\newcommand{\m}{{\mathrm m}}

\usepackage{mytheorems}

\begin{document}

\title{\bfseries Tableaux for Policy Synthesis for MDPs with\\ PCTL*
  Constraints\footnote{This is a longer version of the paper
\textit{Peter Baumgartner, Sylvie Thi\'ebaux, Felipe W. Trevizan:
  Tableaux for Policy Synthesis for MDPs with PCTL* Constraints.
  Proceedings of TABLEAUX 2017, pp. 175--192, Springer 2017},
  \url{https://doi.org/10.1007/978-3-319-66902-1_11}. It contains the proofs of the
  main results and it fixes a bug related to the definition of the semantics of PCTL*
in terms of Markov chains. The semantics is now consistent with the standard
semantics. It required adapting the calculus' $\PQ P$-rule to the new definition.}
  }
\author{Peter Baumgartner, Sylvie Thi\'ebaux, and Felipe Trevizan\\[2ex]
Data61/CSIRO and Research School of Computer Science, ANU, Australia\\[2ex]
{\small Email: \url{first.last@anu.edu.au}}}
\maketitle

\begin{abstract}
Markov decision processes (MDPs) are the standard formalism for
modelling sequential decision making in stochastic environments. Policy
synthesis addresses the problem of how to control or limit the decisions an
agent makes so that a given specification is met. In this paper we
consider PCTL*, the probabilistic counterpart of CTL*, as the specification
language. Because in general the policy synthesis problem for
PCTL* is undecidable, we restrict to policies whose execution history memory
is finitely bounded a priori.
Surprisingly, no algorithm for policy synthesis for this natural and
expressive framework has been developed so far.  We close this gap and
describe a tableau-based algorithm that, given an MDP and a PCTL* specification, derives
in a non-deterministic way a system of (possibly nonlinear) equalities
and inequalities. The solutions of this system, if any, describe the
desired (stochastic) policies.
Our main result in this paper is the correctness of our method, i.e.,
soundness, completeness and termination.  
\end{abstract}

\section{Introduction}
Markov decision processes (MDPs) are the standard formalism for modelling sequential
decision making in stochastic environments, where the effects of an
agent's actions are only probabilistically known. The core problem is to
synthesize a policy prescribing or restricting the actions that the agent may
undertake, so as to guarantee that a given specification is met.  Popular
specification languages for this purpose include CTL, LTL, and their probabilistic
counterparts PCTL and probabilistic LTL (pLTL).  Traditional algorithms for policy
synthesis and probabilistic temporal logic
model-checking~\cite{DBLP:journals/jacm/CourcoubetisY95,DBLP:conf/atva/KwiatkowskaP13}
are based on bottom-up formula
analysis~\cite{DBLP:conf/fsttcs/KuceraS05,DBLP:conf/sfm/KwiatkowskaNP07} or Rabin
automata~\cite{DBLP:conf/ifipTCS/BolligC04,DBLP:journals/tac/DingSBR14,DBLP:conf/cdc/SvorenovaCB13}.

We deviate from this mainstream research in two ways.
The first significant deviation 
is that we consider PCTL* as a specification
language, whereas previous synthesis approaches have been limited to pLTL and PCTL.  PCTL* is the probabilistic counterpart of CTL* and subsumes
both PCTL and pLTL. For example, the PCTL* formula
$\P_{\ge 0.8} \G ((T > 30\degree) \rightarrow \P_{\ge 0.5}  \F \G (T < 24\degree))$ says \textit{``with probability
at least 0.8, whenever the temperature exceeds 30\degree\ it will eventually stay below
24\degree\ with probability at least 0.5''}. Because of the nested probability operator
$\PQ P$ the formula is not in pLTL, and because of the nested temporal operators $\TO F \TO G$
it is not in PCTL either.

Because in its full generality the policy synthesis problem for PCTL* is highly
undecidable~\cite{DBLP:conf/lics/BrazdilBFK06},
one has to make concessions to obtain a decidable fragment. In this
paper we chose to restrict to policies whose execution history memory is
finitely bounded a priori. (For example, policies that choose actions in
the current state dependent on the last ten preceding states.)
However, we do target synthesizing stochastic policies, i.e., the actions are chosen according to a
probability distribution (which generalizes the deterministic case and is known to be needed to satisfy certain formulas \cite{DBLP:conf/ifipTCS/BolligC04}).
Surprisingly, no algorithm for policy synthesis in this somewhat restricted yet 
natural and expressive framework has been developed so far, and this paper closes
this gap.

The second significant deviation from the mainstream is that we pursue a different 
approach based on analytic tableau and mathematical
programming. Our tableau calculus is goal-oriented by focusing on the given
PCTL* formula, which leads to analysing runs only on a by-need basis. This restricts the search
space to partial policies that only cover the states reachable from
the initial state under the policy and for which the formula imposes  
constraints on the actions that can be selected.
In contrast, traditional automata based approaches require a full-blown state space
exploration. (However, we do not have an implementation yet that allows us to evaluate the
practical impact of this.) 
We also believe that our approach, although using somewhat non-standard tableau
features, is conceptually simpler 
and easier to comprehend. Of course, this is rather subjective.

On a high level, the algorithm
works as follows. The input is an MDP, the finite-history component of the policy to
be synthesized,
and a PCTL* formula to be satisfied. Starting from the MDP's initial state, the
tableau calculus symbolically executes the transition system given by the MDP by
analysing the syntactic structure of given PCTL* formula, as usual with tableau
calculi.  Temporal formulas (e.g.,\ $\PQ F \PQ G$-formulas) are expanded repeatedly
using usual expansion laws and trigger state transitions. The process stops at
trivial cases or when a certain loop condition is met.  The underlying loop checking
technique was developed only recently, by Mark Reynolds, in the context of tableau
for satisfiability checking of LTL
formulas~\cite{Reynolds:LTL-Tableaux:GandALF:2016}. It is an essential ingredient of
our approach and we adapted it to our probabilistic setting.

Our tableaux have two kinds of branching. One kind is traditional or-branching,
which represents non-deterministic choice by going down exactly one child node.
It is used, e.g., in conjunction with recursively calling the 
tableau procedure itself. Such calls are necessary to deal with nested $\PQ P$-operators,
since at the time of analyzing a $\PQ P$-formula it is, roughly speaking,
unknown if the formula will hold true under the policy computed only later, as a result
of the algorithm. The
other kind of branching represents a union of alternatives. It is used for
disjunctive formulas and for branching out from a state into successor
states. Intuitively, computing the probability of a disjunctive formula
$\phi_1 \vee \phi_2$ is a function of the probabilities of \emph{both} $\phi_1$ and
$\phi_2$, so both need to be computed. Also, the
probability of an $\TO X$-formula $\TO X \phi$ at a given state is a function of the
probability of $\phi$ at \emph{all} successor states, and so, again, all successor
states need to be considered.

The tableau construction always terminates and derives a system of (possibly
nonlinear) equalities and inequalities over the reals. The solutions of this system, if
any, describe the desired stochastic, finite-history policies. The idea of
representing policies as the solutions of a set of mathematical constraints is
inspired by the abundant work in operations research, artificial intelligence, and
robotics that optimally solves MDPs with simpler constraints using linear programming
\cite{altman:99,DBLP:conf/ijcai/DolgovD05,DBLP:conf/icra/DingPS13,DBLP:conf/aips/TrevizanTSW16}.

Our main result in this paper is the correctness of our algorithm, i.e., 
soundness, completeness and termination. To our knowledge, it is the first and only policy
synthesis algorithm for PCTL* that doesn't restrict the language (but only slightly the policies).

\subsection*{Related Work}
Methods for solving the PCTL* \emph{model checking} problem over
Markov Chains are well established. The (general) policy synthesis however is harder than the model
  checking problem; it is known to be undecidable for even PCTL.
The main procedure works bottom-up from the syntax tree of the given
formula, akin to the standard CTL/CTL* model checking procedure.
Embedded $\PQ P$-formulas are recursively abstracted into boolean variables representing
the sets of states satisfying these formulas, which are computed by LTL model
checking techniques using Rabin automata. Our \emph{synthesis} approach is rather
different. While there is a rough correspondence in terms of recursive calls to treat
$\PQ P$ formulas, we do not need Rabin (or any other) automata; they are supplanted
by the loop-check technique mentioned above.

The work the most closely related to ours is that of 
Br{\'{a}}zdil \emph{et.\
  al.}~\cite{DBLP:conf/stacs/BrazdilKS05,DBLP:conf/concur/BrazdilF07,DBLP:conf/icalp/BrazdilFK08}.
Using B\"uchi automata, they obtain complexity results
depending on the variant of the synthesis problem 
studied. However, they consider only \emph{qualitative} fragments. For the case
of interest in this paper, PCTL*, they obtain results for the fragment
qPCTL*. The logic qPCTL* limits the use of the path quantifier $\PQ P$ to formulas of
the form $\P_{= 1} \psi$ or $\P_{= 0} \psi$, where $\psi$ is a path formula. 
On the other hand, we cover the full logic PCTL* which has arbitrary formulas of the
form $\P_{\sim z} \psi$ where ${\sim} \in \{<,\le,>,\ge\}$ and $z \in [0,1]$.
In contrast to the works mentioned, we have to restrict to memory-dependent policies
with an \emph{a priori} limited finite memory. Otherwise the logic becomes
highly undecidable~\cite{DBLP:conf/lics/BrazdilBFK06}.%

\section{Preliminaries}
\label{sec:preliminaries}
We assume the reader is familiar with basic concepts of Markov Decision Processes
(MDPs), probabilistic model checking, and policy synthesis. See~\cite{DBLP:conf/atva/KwiatkowskaP13,DBLP:conf/sfm/ForejtKNP11,DBLP:books/daglib/0020348}
for introductions and overviews. In the following we summarize the notions
relevant to us and we introduce our notation.

Given a fixed finite vocabulary $AP$ of \define{atomic propositions} $a,b,c,\ldots$ , a
\define{(propositional) interpretation $I$} is any subset of $AP$. It represents the
assignment of each element in $I$ to $\opTrue$ and each other atomic
proposition in $AP \setminus I$ to $\opFalse$. 
A \define{distribution on a countable set $X$} is a
  function $\mu\colon X  \mapsto [0,1]$ such that $\sum_{x\in X}\mu(x) = 1$, and $\Dist(X)$ is the set of all
  distributions on $X$.

A \define{Markov Decision Process (MDP)} is a tuple $\calM = (S, s_\init, A, P, L)$
where:
$S$ is a finite set of states;
$s_\init \in S$ is the \define{initial state};
$A$ is a finite set of \define{actions} and we denote by $A(s) \subseteq A$ the \define{set of
actions enabled} in $s \in S$;
$P(t | s,\alpha)$ is the probability of transitioning to $t \in S$ after applying $\alpha \in A(s)$
in state $s$; and
$L \colon S \mapsto 2^{AP}$ labels each state in $S$ with an interpretation.
We assume that every state has at least one enabled action, i.e., $A(s) \neq \emptyset$
for all $s \in S$, and that $P$ is a distribution on enabled actions, i.e.,
$P(\cdot |s,\alpha) \in \Dist(S)$ iff $\alpha\in A(s)$
and $\Sigma_{t\in S} P(t | s,\alpha) = 0$ iff $\alpha\notin A(s)$.
For any $s$ and $\alpha \in A(s)$ let $\Succ(s,\alpha)  = \{ t \mid P(t | s, \alpha) > 0 \}$ be the states reachable from $s$ with non-zero
probability after applying $\alpha$.


Given a state $s \in S$ of $\calM$, a \define{run from $s$ (of $\calM$)} is an infinite sequence $r = (s = s_1)
\stackrel{\alpha_1}{\longrightarrow} s_2 \stackrel{\alpha_2}{\longrightarrow} s_3 \cdots $ of states
$s_i \in S$ and actions $\alpha_i \in A(s_i)$ such that $P(s_{i+1}|s_i, \alpha_i) > 0$, for all $i
\ge 1$.
%
%
%
%
We denote by \define{$\Runs(s)$} the set of all runs from $s \in S$ and \define{$\Runs = \cup_{s\in
S} \Runs(s)$}.
%
%
A \define{path from $s \in S$ (of $\calM$)} is a finite prefix of a run from $s$ and we define $\Paths(s)$ and
$\Paths$ in analogy to $\Runs(s)$ and $\Runs$. 
We often write runs and paths in abbreviated form as state sequences $s_1 s_2 \cdots$ and
leave the actions implicit. 
Given a path $p = s_1 s_2 \cdots s_n$ let \define{$\first(p) = s_1$} and
\define{$\last(p) = s_n$}.
Similarly, for a run $r = s_1 s_2 \cdots$, \define{$\first(r) = s_1$}
and \define{$r[n] := s_n s_{n+1} \cdots$}, for any $n \ge 1$.

A policy $\pi$ represents a decision rule on how to choose an
action given some information about the environment. 
In its most general form,
  a \define{history-dependent (stochastic) policy (for $\calM$)} is a function
  $\pi\colon \Paths \mapsto \Dist(A)$ such that, for all $p \in \Paths$,
$\pi(p)(\alpha) > 0$ only if $\alpha \in A(\last(p))$. 
%
Technically, the MDP
$\calM$ together with $\pi$ induces
an infinite-state Markov chain
$\calM_{\pi} = (S^{\pi}, s_\init^{\pi}, A, P^{\pi},  L^{\pi})$
over $S^{\pi} = \Paths$ and this way provides
a probability measure $\Pr^{M_\pi}$ on the set of runs of $\calM$ under $\pi$.
See~\cite{kemeny2012denumerable,DBLP:books/daglib/0020348} for details.

However, since $\Paths$ is an infinite set, a history-dependent policy might not be
representable; moreover, the problem of finding such a policy that satisfies PCTL*
constraints is undecidable~\cite{DBLP:conf/lics/BrazdilBFK06}.
We address these issues by working with finite-memory policies.

Finite-memory policies provide a distribution on actions for a current state from $S$ and a current
\emph{mode}. Formally, a
\define{finite-memory policy (for an MDP $\calM$)} is a DFA
$\pi_\fin = (M, \start, \modefun, \act)$ where  $M$ is a finite set of \define{modes},
$\start \colon S \mapsto M$ returns an initial mode to pair with a state $s \in S$,
$\modefun\colon M \times S \mapsto M$ is the mode transition function, and
$\act\colon M \times S \mapsto \Dist(A)$ is the action probability function such
that, for all 
$\pair m s \in M\times S$, $\act(m, s)(\alpha) > 0$ only if $\alpha \in A(s)$. We abbreviate
$\act(m, s)(\alpha)$ as $\act(m,s,\alpha)$. 

Similarly to history-dependent policies above,
a given $\calM$ and $\pi_\fin$ induce a Markov chain denoted as
\define{$\calM_{\pi_\fin} = (S^{\pi_\fin}, s_\init^{\pi_\fin}, A, P^{\pi_\fin},
  L^{\pi_\fin})$}. 
It has the finite state space $S^{\pi_\fin} = M \times S$, initial state
$s_\init^{\pi_\fin} = \pair {\start(s_\init)} {s_\init}$,
set of actions $A$, transition probability function
$P^{\pi_\fin}(\pair {m'} {s'} | \pair m s ) := \Sigma_{\alpha \in A(s)}\,\act(m,s,\alpha) \cdot P(s'
| s,\alpha)$ if $m' = \Delta(m,s)$ and 0 otherwise, and labelling function
$L^{\pi_\fin}(\pair m s) := L(s)$.
A \define{run from
  $s_1$ of $\calM_{\pi_\fin}$} is a sequence of the form
$\pair {m_1} {s_1} \pair {m_2} {s_2} \cdots $ such that
$m_1 = \start(s_1)$, 
$m_{i+1} = \Delta(m_i,s_i)$ and
$P^{\pi_\fin}(\pair {m_{i+1}} {s_{i+1}} | \pair {m_i} {s_i}) > 0$, for all
$i \ge 1$. 
Notice that any run from $s_1$ of $\calM_{\pi_\fin}$ satisfies
$\act(m_i,s_i,\alpha_i) \cdot P(s_{i+1} | s_i,\alpha_i) > 0$  for some $\alpha_i \in A(s_i)$
and hence  induces a run
$s_1 \stackrel{\alpha_1}{\longrightarrow} s_2 \stackrel{\alpha_2}{\longrightarrow} s_3\cdots$ from
$s_1$ of $\calM$.
The notions of $\Runs$, $\Paths$, $\first$, $\last$ etc apply to Markov chains as well. 
For instance, $\first(\pair {m_1} {s_1} \pair {m_2} {s_2} \cdots ) = \pair {m_1} {s_1}$.
Let $\Pr^{\calM_{\pi_\fin}}$ denote the probability measure on the set of runs of $\calM_{\pi_\fin}$.

Finite-memory policies are more expressive than Markovian policies. For instance, the
$\Delta$-component of $\pi_\fin$ can be 
setup, e.g., to encode in $\pair m s$ ``the last ten states preceding $s$''.

\subsection{Policy Synthesis for PCTL*}
\define{(PCTL*) formulas} follow the following grammar: 
\begin{align}
  \phi\quad & := \quad \opTrue \mid a \in AP \mid \phi \wedge \phi \mid \neg\phi \mid \P_{\sim z} \psi \tag{State formula}\\
  \psi\quad & := \quad \phi \mid \psi  \wedge \psi \mid \neg\psi \mid \X \psi \mid \psi \U \psi \tag{Path formula}
\end{align}
In the definition of state formulas,
${\sim} \in \{ <, \leq, >, \geq \}$ and $0 \leq z \leq 1$. A \define{proper path formula} is a path
formula that is not a state formula.
A formula is \define{classical} iff it is made from atomic propositions and the 
Boolean connectives $\neg$ and $\land$ only (no occurrences of $\PQ P$,
$\TO X$ or $\TO U$). 
We write $\opFalse$ as a shorthand for $\neg \opTrue$.

We are going to define the semantics of PCTL* in terms of Markov chains in the
usual way. That is, given an MDP 
$\calM = (S, s_\init, A, P, L)$ and a policy $\pi$, we fix the Markov chain
$\calM_\pi = (S^\pi, s_\init^\pi, A, P^\pi, L^\pi)$ as described above.
In fact, in our case the policies of interest will  always be 
 finite-memory policies.

Let $\Runs^{\calM_\pi}$ denote the set of runs of $\calM_\pi$.
For state formulas~$\phi$ and states $s \in S^\pi$ define the satisfaction relation
\define{$\calM_\pi, s \models \phi$}, briefly \define{$s \models \phi$}, as follows:
  \begin{align*}
    s & \models \opTrue &  s &\models \phi_1 \land \phi_2 \text{\  iff\  $s \models \phi_1$ and $s \models \phi_2$ } \\
    s & \models a \text{\  iff\   $a\in L(s)$} &
    s &\models \neg\phi \text{\  iff\  $s \not\models \phi$ } \\
    s &\models \P_{\sim z} \psi \text{\  iff\ } 
          \mathop{\Pr^{\calM_\pi}}(\{r \in \textstyle{\Runs^{\calM_\pi}}(s) \mid \calM_\pi, r \models \psi\}) \sim z 
        \span\span
  \end{align*}
For path formulas $\psi$ and runs $r \in \Runs^{\calM_\pi}$ 
define the satisfaction relation \define{$\calM_\pi, r \models \psi$}, briefly \define{$r \models \psi$} as follows:
  \begin{align*}
    r & \models \phi \text{\  iff\  $\first(r) \models \phi$} & 
    r & \models \psi_1 \land \psi_2 \text{\  iff\  $r \models \psi_1$ and $r \models \psi_2$ }\\
    r & \models \neg\psi \text{\  iff\  $r \not\models \psi$} & 
    r & \models \X \psi \text{\  iff\  $r[2] \models \psi$}\\ 
    r & \models \psi_1 \U \psi_2 \text{\  iff\  exists $n \ge 1$ s.t.\ $r[n] \models \psi_2$
               and  $r[m] \models \psi_1$ for all $1 \le m < n$ }\span\span
  \end{align*}

  We say that an MDP $\calM$ and a policy $\pi$ satisfies $\phi$, written as
    $\calM,\pi \models \phi$ iff $\calM_\pi, s_\init^\pi \models \phi$.

    In this paper we focus on the problem of synthesizing only the $\act$-component of an otherwise fully
specified finite memory policy. More formally:
\begin{definition}[Policy Synthesis Problem]
    \label{def:policy-synthesis-problem}
    Let $\calM = (S, s_\init, A, P, L)$ be an MDP, and
    ${\pi_\fin} = (M, \start, \Delta, \cdot)$ be
    a partially specified finite-memory policy with $\act$ unspecified. 
    Given state formula $\phi$, 
    find $\act$ s.th.\ $\calM, {\pi_\fin} \models \phi$ if it exists, otherwise report failure. 
\end{definition}

\subsection{Useful Facts About  PCTL* Operators}
Next we summarize some well-known or easy-to-prove facts about PCTL* operators.
By the \define{expansion laws} for the $\PQ U$-operator we mean the following equivalences:
\begin{align}
  \psi_1 \U \psi_2 & \equiv \psi_2 \vee (\psi_1 \wedge \X (\psi_1 \U \psi_2 )) &
  \neg(\psi_1 \U \psi_2) & \equiv \neg\psi_2 \land (\neg \psi_1 \vee \X \neg(\psi_1 \U \psi_2 )) \label{eq:e} \tag{E}
\end{align}

For ${\sim} \in \{ <, \leq, >, \geq \}$ define the operators $\compl{\sim}$ and $[\sim]$ as follows:
\begin{align*}
{\compl{<}} & ={\ge} &
{\compl{\leq}} & ={>} &
{\compl{>}} & ={\leq} &
{\compl{\ge}} & ={<}  &
[<] & ={>} &
[\leq] & ={\ge} &
[>] & ={<} &
[\ge] & ={\le} 
\end{align*}

Some of the following equivalences cannot be used for ``model checking'' PCTL* (the
left~(\ref{eq:p1}) equivalence, to be specific) where actions are implicitly
universally quantified. In the context of Markov Chains, which we implicitly have,
there is no problem:
\begin{align}
  \neg \P_{\sim z} \psi & \equiv \P_{\compl{\sim} z} \psi &
  \P_{\sim z} \neg \psi & \equiv \P_{\mathop{[\sim]} 1-z} \psi \label{eq:p1} \tag{P1} \\
  \P_{\ge 0} \psi & \equiv \opTrue & \P_{> 1} \psi & \equiv \opFalse \label{eq:p2} \tag{P2}\\
  \P_{\le 1} \psi  & \equiv \opTrue & \P_{< 0} \psi & \equiv \opFalse \label{eq:p3} \tag{P3}\\
  \P_{\geq u} \P_{\sim z} \psi & \equiv \P_{\sim z} \psi \quad\text{if $u \neq 0$} &
  \P_{> u} \P_{\sim z} \psi & \equiv \P_{\sim z} \psi \quad\text{if $u\neq 1$} \label{eq:p4} \tag{P4}\\
  \P_{\leq u} \P_{\sim z} \psi & \equiv \P_{\ge 1-u}\P_{\compl{\sim} z} \psi &
  \P_{< u} \P_{\sim z} \psi & \equiv \P_{> 1-u}\P_{\compl{\sim} z} \psi \label{eq:p5} \tag{P5}
\end{align}

\subsection{Nonlinear Programs}
Finally, a \define{(nonlinear) program} is a set $\Gamma$ of constraints of the form $e_1 \bowtie e_2$ where
${\bowtie} \in \{ <, \le, >, \geq, \doteq \}$ and $e_1$ and $e_2$ are arithmetic expressions
comprised of numeric real constants and variables.  The numeric operators are $\{+, -, \cdot, /\}$,
all with their expected meaning (the symbol $\doteq$ is equality).
All variables are implicitly bounded over the range $[0,1]$.
A solver (for nonlinear programs) is a decision procedure that returns a satisfying variable
assignment (a solution) for a given $\Gamma$, and reports unsatisfiability if no solution exists. We
do not further discuss solvers in the rest of this paper, we just assume one as given. Examples of
open source solvers include Ipopt and Couenne.\footnote{\url{http://projects.coin-or.org/}.}

\section{Tableau Calculus}
\label{sec:tableau}
\subsection{Introduction and Overview}
We describe a tableau based algorithm for the policy synthesis problem in
Definition~\ref{def:policy-synthesis-problem}. Hence assume as given an MDP
$\calM = (S, s_\init, A, P, L)$ and a partially specified finite-memory policy
${\pi_\fin} = (M, \start, \Delta, \cdot)$ with $\act$ unspecified.

A \define{labelled formula $\calF$} is of the form $\pair m s:\Psi$ where $\pair m s \in M\times S$ and
$\Psi$ is a possibly empty set of path formulas, interpreted
conjunctively.
When we speak of the \define{probability of $\pair m s:\Psi$} we mean the value of
$\Pr^{\calM_{\pi_\fin}}(\{r \in \Runs^{\calM_{\pi_\fin}}(\pair m s) \mid \calM_{\pi_\fin}, r \models \bigwedge \Psi\})$
for the completed $\pi_\fin$.
For simplicity we also call $\Psi$ a ``formula'' and call $\pair m s$ a \define{policy state}.
A \define{sequent}
is an expression of the form $\Gamma \vdash \calF$ where $\Gamma$ is a program. 

Our algorithm consists of three steps, the first one of which is a tableau construction.
A \define{tableau for $\Gamma \vdash \calF$} is a finite tree whose root is labelled with
$\Gamma \vdash \calF$ and such that every inner node is labelled with the premise of an inference
rule and its children are labelled with the conclusions, in order. If $\Gamma \vdash \calF$ is
the label of an inner node we call $\calF$ the \define{pivot of the node/sequent/inference}.
By a \define{derivation from $\Gamma \vdash \calF$}, denoted
by \define{$\TABLEAU(\Gamma \vdash  \calF)$},
we mean any tableau for $\Gamma \vdash \calF$ obtained by stepwise construction, starting
from a root-node only tree and applying an inference rule to (the leaf of)  every branch as long as possible.
There is one inference rule, the \IR{$P$}-rule, which recursively calls the algorithm itself.
A branch is terminated when no inference rule is applicable, which is exactly the
case when its leaf is labelled by a pseudo-sequent, detailed below. The inference
rules can be applied in any way, subject to only preference constraints.

Given a state formula $\phi$, the algorithm starts with a derivation from
$\Gamma_\init \vdash \calF_\init\ :=\ 
\{x_{\pair {\start(s_\init)} {s_\init}}^{\{\phi\}}\doteq 1\} \vdash
\pair {\start(s_\init)} {s_\init}:\{\phi\}$. (The constraint
$\Gamma_\init$ forces $\phi$ to be ``true''.)  The derivation represents the obligation to derive
a satisfiable extension $\Gamma_\final \supseteq \Gamma_\init$ whose solutions $\sigma$ determine 
the $\act$-component $\act_\sigma$ of $\pi_\fin$ such that
$\calM_{\pi_\fin}, \pair {\start(s_\init)} {s_\init} \models \phi$.
In more detail, $\Gamma_\final$ will contain constraints of the form
$x_{\pair m s}^\alpha \doteq 0$ or $x_{\pair m s}^\alpha > 0$ for the probability of applying
action $\alpha$ in policy state $\pair m s$. Let the \define{policy domain of a program $\Gamma$} be
the set of all policy states $\pair {m} {s} \in M \times S$ such that $x_{\pair m s}^\alpha$ occurs
in $\Gamma$, for some $\alpha$.
This lets us initially define $\act_\sigma(m,s,\alpha) := \sigma(x_{\pair m s}^\alpha)$ for every $\pair m s$
in the policy domain of $\Gamma_\final$. Only for the purpose of satisfying the definition
of finite memory policies, we then make $\act_\sigma$ trivially total by choosing an \emph{arbitrary} distribution
for $\act_\sigma(m,s)$ for all remaining $\pair m s \in M\times S$. (The latter 
are not reachable and hence do not matter.) 
We call \define{$\pi_\fin(\sigma) := (M, \start, \Delta, \act_\sigma)$} the \define{policy completed by $\sigma$}.
%

Similarly, $\Gamma_\final$ contains variables of the form $x_{\pair m s}^\Psi$, and
$\sigma(x_{\pair m s}^\Psi)$ is
the probability of ${\pair m s}:\Psi$ under the policy $\pi_\fin(\sigma)$.
(We
actually need these variable indexed by tableau nodes, see below.)
If $\Psi$
is a state formula  
its value will be 0 or 1, \mbox{encoding truth values.}


Contrary to traditional tableau calculi, the result of the computation -- the extension
$\Gamma_\final$ -- cannot always be obtained in a branch-local way. To explain,
there are two kinds of
branching in our tableaux: \define{don't-know (non-deterministic) branching} and \define{union branching}. The former
is always used for exhaustive case analysis, e.g., whether $x_{\pair m s}^\alpha \doteq 0$ or
$x_{\pair m s}^\alpha > 0$, and the algorithm guesses which alternative to take (cf.\ step
2 below). The latter analyzes the Boolean
structure of the pivot. Unlike as with traditional tableaux, \emph{all} children need to
be expanded, and each fully expanded branch contributes to $\Gamma_\final$.

More precisely, we formalize the synthesis algorithm as a three-step procedure.
\define{Step one} consists in deriving $\TABLEAU(\Gamma_\init \vdash \calF_\init)$.
\define{Step two} consists in removing from the step one tableau every don't-know branching by
retaining exactly one child of the parent node of the don't-know branching, and
deleting all other children and the subtrees below them. This itself is a don't-know
non-deterministic process; it corresponds to going down one branch in traditional
tableau. The result is denoted by $\CHOOSE(T_1)$, where $T_1$
is the step one tableau.
\define{Step three} consists in first building a combined program by taking the
union of the $\Gamma$'s in the leaves of the branches of the step two tableau.  This
program then is extended with a set of constraints by the $\FORCE$ operator. More precisely,
$\FORCE$ing captures the situation when a run reaches a bottom strongly connected
component (BSCC). Any formula is satisfied in a BSCC with probability 0 or 1,  
which can be determined solely by qualitative formula evaluation in the BSCC.
Details are below. For now let us just define
\define{$\GAMMA(T_2) =  \bigcup\, \{ \Gamma \mid \text{$\Gamma\vdash \cdot$ is the leaf of a branch in
$T_2$} \} \cup \FORCE(T_2)$ where $T_2 = \CHOOSE(T_1)$}.

Our main results are the following. See Appendix~\ref{sec:proofs} for proofs.
\begin{theorem}[Soundness]
  \label{th:soundness}
  Let $\calM = (S, s_\init, A, P, L)$ be an MDP,
  $\pi_\fin = (M, \start, \Delta, \cdot)$ be
  a partially specified finite-memory policy with $\act$ unspecified,
  and $\phi$ a state formula. Suppose there is a program
  $\Gamma_\final := \GAMMA(\CHOOSE(\TABLEAU(
\{x_{\pair {\start(s_\init)} {s_\init}}^{\{\phi\}}\doteq 1\} \vdash \pair {\start(s_\init)} {s_\init}:\{\phi\})))$
  such that $\Gamma_\final$ is satisfiable.
  Let $\sigma$ be any solution of $\Gamma_\final$ and $\pi_\fin(\sigma)$ be the policy completed by $\sigma$.
  Then it holds $\calM, \pi_\fin(\sigma) \models \phi$.
\end{theorem}

\begin{theorem}[Completeness]
  \label{th:completeness}
  Let $\calM = (S, s_\init, A, P, L)$ be an MDP,
  $\pi_\fin = (M, \start, \Delta, \act)$ 
  a finite-memory policy, and $\phi$ a state formula. Suppose
  $\calM, \pi_\fin \models \phi$. Then there is a satisfiable program
  $\Gamma_\final := \GAMMA(\CHOOSE(\TABLEAU(\{x_{\pair {\start(s_\init)}
    {s_\init}}^{\{\phi\}}\doteq 1\} \vdash \pair {\start(s_\init)} {s_\init}:\{\phi\})))$
  and a solution $\sigma$  of $\Gamma_\final$ such that $\act_\sigma(m,s,\alpha) = \act(m,s,\alpha)$
  for every pair $\pair m s$ in the policy domain of $\Gamma_\final$.
  Moreover $\calM, \pi_\fin(\sigma) \models \phi$.
\end{theorem}

\subsection{Inference Rules}
There are two kinds of inference rules, giving two kinds of branching:
{
\begin{gather*}
  \infrule[\textit{Name}]{
     \Gamma \vdash \pair{m}{s}: \Psi
  }{
    \Gamma_\text{\textit{left}} \vdash \pair{m}{s}: \Psi \qquad \Gamma_\text{\textit{right}} \vdash \pair{m}{s}: \Psi
  } \infrulecond{\text{if \textit{condition}}}
  \tag{Don't-know branching}
\end{gather*}}\ignorespaces
The pivot in the premise is always carried over into both conclusions. Only the constraint $\Gamma$
is modified into $\Gamma_\text{\textit{left}} \supseteq \Gamma$ and $\Gamma_\text{\textit{right}} \supseteq \Gamma$,
respectively, for an exhaustive case analysis.
{
\begin{gather*}
\infrule[\textit{Name}]{
    \Gamma  \vdash \pair{m}{s}: \Psi
  }{
    \Gamma_1 \vdash \pair{m_1}{s_1}: \Psi_1  \quad \cup \quad \cdots \quad \cup \quad
      \Gamma_n \vdash \pair{m_n}{s_n}: \Psi_n
  }\infrulecond{\text{if \textit{condition}}\  (n \ge 1)}
  \tag{Union branching}
\end{gather*}}\ignorespaces
All union branching rules satisfy
$\Gamma_i \supseteq \Gamma$, and $\pair{m_i}{s_i} = \pair{m}{s}$ or
$\pair{m_i}{s_i} = \pair{\Delta(m,s)}{t}$ for some state $t$.
The $\cup$-symbol is decoration for distinguishing the two kinds of
branching but has no meaning beyond that.
Union branching stands for the union of the runs from $\pair{m_i}{s_i}$ satisfying
$\Psi_i$, and computing its probability requires to develop
\emph{all} $n$ child nodes.

We need to clarify a technical add-on.
Let $u$ be the tableau node with the premise pivot $\pair{m}{s}: \Psi$.
An union branching inference extends $u$ with children nodes, say, $u_1,\ldots,u_n$, with conclusion pivots
$\pair{m_i}{s_i}: \Psi_i$. The program $\Gamma_n$ will contain a constraint that makes a
variable $(x_u)_{\pair{m}{s}}^\Psi$ for the premise dependent on all variables
$(x_{u_i})_{\pair{m_i}{s_i}}^{\Psi_i}$ for the
respective conclusions. \emph{This is a key invariant and is preserved by
all inference rules.}
In order to lighten the notation, however, we usually drop the
variable's index, leaving the node implicit. For instance, we write
$x_{\pair{m}{s}}^\Psi$ instead of $(x_u)_{\pair{m}{s}}^\Psi$.  The
  index $u$ is needed for not
inadvertently identifying the same pivot at different points in the symbolic
execution of a run. Fresh names $x,y,z,\ldots$ for the
variables would do as well.

Most unary union branching rules have a premise $\Gamma \vdash \pair{m}{s}:  \{\psi \} \uplus \Psi$ and
the conclusion is $\Gamma, \gammaone \vdash \pair{m}{s}:  \Psi'$, for some $\Psi'$.
The pivot is specified by pattern matching, where $\uplus$ is disjoint union, and
$\gammaone$ is a macro that expands to
 $x_{\pair{m}{s}}^{\{\psi \} \uplus \Psi} \doteq x_{\pair{m}{s}}^{\Psi'}$.


Other inference rules derive pseudo-sequents of the form
$\Gamma \vdash \Cross$,  $\Gamma \vdash \Checkmark$, $\Gamma \vdash \textsf{Yes-Loop}$ and $\Gamma \vdash \textsf{No-Loop}$.
They indicate that the probability of the pivot is 0, 1, or that a loop situation
arises that may need further analysis. Pseudo-sequents are always leaves.

Now we turn to the concrete rules. They are listed in decreasing order
of preference.
\paragraph{Rules for classical formulas}
{
  \begin{gather*}
\infrule[$\top$]{
\Gamma \vdash \pair{m}{s}:  \{\psi \} \uplus \Psi
}{ 
\Gamma, \gammaone \vdash \pair{m}{s}:  \Psi
} \infrulecond{\left\{\pbox{ if $\psi$ is clas-\\ sical and\\ $L(s) \models \psi$}\right.}
\quad\infrule[\Cross]{
\Gamma \vdash \pair{m}{s}:  \{\psi\} \uplus \Psi
}{ 
\Gamma, x_{\pair{m}{s}}^{\{\psi\}\uplus\Psi} \doteq 0 \vdash \Cross
} \infrulecond{\left\{\pbox{ if $\psi$ is clas-\\ sical and\\ $L(s) \not\models \psi$}\right.}
\\[0.1ex]
\infrule[\Checkmark]{
\Gamma \vdash \pair{m}{s}: \emptyset 
}{ 
\Gamma, x_{\pair{m}{s}}^\emptyset \doteq 1  \vdash \Checkmark
} \qquad
\infrule[$\neg\neg$]{
\Gamma \vdash \pair{m}{s}:  \{\neg\neg\psi\}\uplus \Psi
}{ 
\Gamma, \gammaone \vdash \pair{m}{s}:  \{\psi\} \cup \Psi
} \\[0.1ex]
\infrule[$\neg\PQ P$]{
\Gamma \vdash \pair{m}{s}:  \{\neg \P_{\sim z} \psi\} \uplus \Psi
}{ 
\Gamma, \gammaone \vdash \pair{m}{s}:  \{\P_{\compl{\sim} z} \psi\} \cup \Psi
}\qquad
\infrule[$\PQ P\neg$]{
\Gamma \vdash \pair{m}{s}:  \{ \P_{\sim z} \neg \psi\} \uplus \Psi
}{ 
\Gamma, \gammaone \vdash \pair{m}{s}:  \{\P_{\mathop{[\sim]} 1-z} \psi\} \cup \Psi
} 
\end{gather*}}
These are rules for evaluating classical formulas and for negation.
The \IR{$\Cross$} rule terminates the branch and assigns a probability of 0 to the premise pivot, as no run
from $\pair m s$ satisfies (the conjunction of) $\{\psi\} \uplus \Psi$, as $\psi$ is false in
$s$.  A similar reasoning applies to the \IR{$\top$} and \IR{$\Checkmark$} rules.
The $\neg \PQ P$ and $\PQ P\neg$ rules are justified by law (\ref{eq:p1}).
The $\PQ P\neg$ rule is needed for removing negation between
$\PQ P$-formulas as in
$\P_{\sim z} \neg \P_{\sim v} \psi$. 
%
%
%
%
\paragraph{Rules for conjunctions}
{
\begin{gather*}
\infrule[$\land$]{
\Gamma \vdash \pair{m}{s}:  \{\psi_1 \land \psi_2\} \uplus \Psi
}{ 
\Gamma, \gammaone \vdash \pair{m}{s}:  \{\psi_1, \psi_2\} \cup \Psi
} 
\\[0.1ex]
\infrule[$\neg\land$]{
\Gamma \vdash \pair{m}{s}: \{\neg(\psi_1 \land \psi_2)\} \uplus \Psi
}{ 
\Gamma \vdash \pair{m}{s}: \{\neg\psi_1\} \cup \Psi \quad \cup \quad \Gamma, \gamma \vdash \pair{m}{s}: \{\psi_1,\neg\psi_2\}\cup \Psi
}\\[0.1ex]
\text{where $\gamma\ =\  x_{\pair{m}{s}}^{\{\neg(\psi_1 \land \psi_2)\} \uplus \Psi } \doteq
x_{\pair{m}{s}}^{\{\neg\psi_1\} \cup \Psi} + x_{\pair{m}{s}}^{\{\psi_1,\neg\psi_2\}\cup \Psi}$}
\end{gather*}}
These are rules for conjunction. Not both $\psi_1$ and $\psi_2$ can be classical by
preference of the \IR{$\top$} and \IR{$\Cross$} rules.
The \IR{$\land$} rule is obvious with the conjunctive reading of formula sets.
The \IR{$\neg\land$} rule deals, essentially, with the disjunction $\neg\psi_1 \vee \neg\psi_2$,
which requires splitting. However, unlike to the classical logic case, $\neg\psi_1 \vee \neg\psi_2$
represents the union of the runs from $s$ satisfying $\neg\psi_1$ and the runs
from $s$ satisfying $\neg\psi_2$. As these sets may overlap the rule
works with a \emph{disjoint} union by taking $\neg\psi_1$ on
the one side, and $\psi_1 \land \neg\psi_2$ on the other side so that it is correct to add their
probabilities up in $\gamma$. 
%
%
%
\paragraph{Rule for simplification of $\PQ P$-formulas}
{
\begin{gather*}
\infrule[$\PQ P1$]{
\Gamma \vdash \pair{m}{s}:  \{\P_{\sim z} \psi\}\uplus \Psi
}{ 
\Gamma, \gammaone \vdash \pair{m}{s}:  \{\psi'\} \cup \Psi
} \infrulecond{\left\{\pbox{if $\P_{\sim z} \psi$ is the left hand side of an equivalence \\
      (\ref{eq:p2})-(\ref{eq:p5})  and $\psi'$ is its right hand side}\right.}
\\[0.1ex]
\infrule[$\PQ P2$]{
\Gamma \vdash \pair{m}{s}:  \{\P_{\sim z} \psi \} \uplus \Psi
}{ 
\Gamma, \gammaone \vdash \pair{m}{s}:  \{\psi\} \cup \Psi
} \infrulecond{\text{ if \emph{see text}}}
\\[0.1ex]
\infrule[$\PQ P3$]{
\Gamma \vdash \pair{m}{s}:  \{\P_{\sim z} \psi\} \uplus \Psi
}{ 
\Gamma, \gammaone \vdash \pair{m}{s}:  \{\neg\psi\} \cup \Psi
} \infrulecond{\text{ if \emph{see text}}}
\end{gather*}}
These are rules for simplifying $\PQ P$-formulas.
The condition in \IR{$\PQ P2$} is ``${\sim} \in \{>, \geq\}$ and  $\psi$ is a state formula'', and
in \IR{$\PQ P3$} it is ``${\sim} \in \{<, \leq\}$ and  $\psi$ is a state formula''.
In the rules \IR{$\PQ P2$} and \IR{$\PQ P3$} trivial cases for $z$ are
excluded by preference of \IR{$\PQ P1$}. Indeed, this preference is even needed for soundness.

The rule \IR{$\PQ P2$} can be explained as follows: suppose we want to know if
$\calM_{\pi_\fin}, \pair m s \models \P_{\sim z} \psi$.
For that we need the probability of the set of runs from
$\pair m s$ that satisfy $\psi$ and compare it with $z$. Because $\psi$ is a \emph{state}
formula this set is comprised of all runs from $s$ if $\calM_{\pi_\fin}, \pair m s \models \psi$, or the empty 
set otherwise, giving it probability 1 or 0, respectively.
With ${\sim} \in \{>, \ge\}$ conclude  $\calM_{\pi_\fin}, \pair m s \models \P_{\sim z} \psi$, or its negation, respectively.
The rule \IR{$\PQ P3$} is justified analogously. The only difference is that ${\sim} \in
\{<,\le\}$ and so the $\P_{\sim z}$ quantifier acts as a negation operator
instead of idempotency.

At this stage, when all rules above have been applied exhaustively to a given branch,
the leaf of that branch must be of the
form $\Gamma   \vdash \pair{m}{s}: \{ \P_{\sim z_1} \psi_1, \ldots, \P_{\sim z_n} \psi_n \}$, for some $n \geq 0$,
where each $\psi_i$ is a non-negated proper path formula.
\paragraph{Rules for $\PQ P$-formulas}
{
  \begin{gather*}
\infrule[$\PQ P$]{
\Gamma \vdash \pair{m}{s}:  \Psi
}{ 
\Gamma,\Gamma',\gammaleft   \vdash \pair{m}{s}:  \Psi \qquad
\Gamma,\Gamma',\gammaright   \vdash \pair{m}{s}:  \Psi
}\infrulecond{\left\{\pbox{ if $\P_{\sim z} \psi\ \in\ \Psi$, and \\$\gammaleft \notin \Gamma$ and $\gammaright \notin \Gamma$}\right.}\\[1ex]
\infrule[$\PQ{P} \top$]{
\Gamma \vdash \pair{m}{s}:  \{\P_{\sim z} \psi\} \uplus \Psi
}{ 
\Gamma, \gammaone \vdash \pair{m}{s}:  \Psi
} \infrulecond{\text{ if $\gammaleft  \ \in\ \Gamma$}}
\\[0.1ex]
  \infrule[$\PQ{P}\Cross$]{
\Gamma \vdash \pair{m}{s}:  \{\P_{\sim z} \psi\} \uplus \Psi
}{ 
\Gamma, x_{\pair{m}{s}}^{\{\P_{\sim z} \psi\} \uplus \Psi} \doteq 0 \vdash \Cross
} \infrulecond{\text{ if $\gammaright \ \in\ \Gamma$}}\\
\text{where }
\begin{aligned}[t]
  \Gamma' & = \GAMMA(\CHOOSE(\TABLEAU(\emptyset \vdash \pair{m}{s}: \{\psi\})))\text{,} \\
  \gammaleft & =  x_{\pair{m}{s}}^{\{\psi\}} \sim z \text{, and }
  \gammaright =  x_{\pair{m}{s}}^{\{\psi\}} \compl{\sim} z 
\end{aligned}
\end{gather*}}\ignorespaces
Unlike classical formulas, $\PQ P$-formulas cannot be evaluated in a
state, because their truth value depends on the solution of the program $\Gamma_\final$.
The \IR{$\PQ P$}  rule analyzes $\P_{\sim z} \psi$ in a deferred way by first getting a constraint
$x_{\pair{m}{s}}^{\{\psi\}} \doteq e$, for some expression $e$,
for the probability of $\pair{m}{s}: \{\psi\}$ by a
recursive call.
This call is not needed if $\Gamma$ already determines a truth value for
$\P_{\sim z} \psi$ because of $\gammaleft \in \Gamma$ or $\gammaright \in \Gamma$.
These tests are done modulo node labels of variables, i.e.,
$(x_u)_{\pair{m}{s}}^{\{\psi\}}$ and
$(x_v)_{\pair{m}{s}}^{\{\psi\}}$ are considered equal for any $u,v$.
Because the value of $e$ is not known at the time of the
inference, the $\PQ P$ rule don't-know
non-deterministically branches out into whether $x_{\pair{m}{s}}^{\{\psi\}} \sim z$ holds or not, as per the 
constraints $\gammaleft$ and $\gammaright$.
The \IR{$\PQ{P} \top$} and \IR{$\PQ{P}\Cross$} rules then lift the corresponding case to the
evaluation of $\P_{\sim z} \psi$, which is possible now thanks to
$\gammaleft$ or $\gammaright$. 

Observe the analogy between these
rules and their counterparts \IR{$\top$} and \IR{$\Cross$} for classical formulas.
Note that the rules \IR{$\PQ P$}, \IR{$\PQ P\top$} and \IR{$\PQ P\Cross$} cannot be combined into
one, because $\gammaleft$ or $\gammaright$ could have been added earlier,
further above in the branch, or in a recursive call.  In this case only
\IR{$\PQ{P} \top$}/\IR{$\PQ{P}\Cross$} can applied.

At this stage, in a leaf $\Gamma \vdash \pair{m}{s}:\Psi$ the set $\Psi$ cannot contain any 
state formulas, as  they would all be eliminated by the inference rules
above; all formulas in $\Psi$ now are possibly negated $\TO X$-formulas or $\TO U$-formulas.
\paragraph{Rules for $\PQ U$-formulas}
{
\begin{gather*}
\infrule[$\TO U$]{
\Gamma \vdash \pair{m}{s}: \{\psi_1 \U \psi_2\} \uplus \Psi
}{ 
\Gamma \vdash \pair{m}{s}: \{ \psi_2 \} \cup \Psi \qquad \cup \qquad \Gamma,\gamma \vdash \pair{m}{s}:\{ \psi_1, \neg \psi_2,\, \X (\psi_1 \U \psi_2)\} \cup \Psi
}
\\[0.1ex]
\text{where $\gamma\ =\ x_{\pair{m}{s}}^{\{\psi_1 \U \psi_2\} \uplus \Psi} \doteq x_{\pair{m}{s}}^{\{\psi_2\} \cup \Psi} + x_{\pair{m}{s}}^{\{\psi_1, \neg\psi_2, \X (\psi_1 \U \psi_2) \} \cup \Psi}$}
\\[2ex]
\infrule[$\neg\TO U$]{
\Gamma \vdash \pair{m}{s}: \{\neg(\psi_1 \U \psi_2) \} \uplus \Psi
}{ 
\Gamma  \vdash \pair{m}{s}: \{ \neg\psi_1, \neg\psi_2 \} \cup \Psi\quad\cup\quad \Gamma,\gamma \vdash \pair{m}{s}: \{\psi_1, \neg \psi_2, \X \neg(\psi_1 \U \psi_2) \} \cup \Psi
}
\\[0.1ex]
\text{where $\gamma\ =\ x_{\pair{m}{s}}^{\{\neg(\psi_1 \U \psi_2)\} \uplus \Psi} \doteq x_{\pair{m}{s}}^{\{\neg\psi_1,\neg\psi_2\} \cup \Psi} + x_{\pair{m}{s}}^{\{\psi_1, \neg\psi_2, \X \neg(\psi_1 \U \psi_2) \} \cup \Psi}$}
\end{gather*}}\ignorespaces
These are expansion rules for $\TO U$-formulas.
The standard expansion law is $\psi_1 \U \psi_2  \equiv \psi_2 \vee (\psi_1 \wedge \X (\psi_1 \U \psi_2 ))$. As with the
\IR{$\neg\land$} rule, the disjunction in the expanded formula needs to be disjoint
by taking 
$\psi_2 \vee (\psi_1 \land \neg \psi_2 \wedge \X (\psi_1 \U \psi_2 ))$ instead. Similarly for \IR{$\neg\TO U$}.
\paragraph{Rule for $\neg\TO X$}
{
\begin{gather*}
\infrule[$\neg\TO X$]{
\Gamma \vdash \pair{m}{s} :  \{\neg \X \psi\} \uplus \Psi 
}{ 
\Gamma,\gammaone \vdash \pair{m}{s} :  \{ \X \neg\psi \} \cup \Psi
}
\end{gather*}}\ignorespaces
The \IR{$\neg \TO X$} rule is obvious.

At this stage, if $\Gamma \vdash \pair{m}{s}:\Psi$ is a leaf
sequent then $\Psi$ is of the form $\{\X \psi_1, \ldots, \X \psi_n \}$, for some
$n \ge 1$.  This is an important configuration that justifies a name: we say that a
labelled formula $\pair{m}{s}:\Psi$, a sequent $\Gamma \vdash \pair{m}{s}:\Psi$ or a node labelled with
$\Gamma \vdash \pair{m}{s}:\Psi$ is \define{poised} if $\Psi$ is of the form
$\{\X \psi_1,\ldots, \X \psi_n\}$ where $n \ge 1$. (The notion ``poised'' is taken
  from~\cite{Reynolds:LTL-Tableaux:GandALF:2016}.)
A poised $\pair{m}{s} : \{\X \psi_1, \ldots, \X \psi_n \}$ will be expanded by transition
into the successor states of $s$ by using enabled actions $\alpha \in A(s)$. 
That some $\alpha$ is enabled does not, however, preclude a
policy with $\act_\sigma(m,s,\alpha) = 0$. The rule \IR{A} makes a guess whether this is the case
or not:
%
%
%
\paragraph{Rules for prescribing actions}
{
\begin{gather*}
 \infrule[A]{
\Gamma \vdash \pair{m}{s}: \Psi
}{ 
\Gamma,\gammaleft \vdash \pair{m}{s}: \Psi \qquad \Gamma,\gammaright \vdash \pair{m}{s}: \Psi
} \infrulecond{\left\{\pbox{ if $\Gamma \vdash \pair{m}{s}:\Psi$ is poised, \\
      $\alpha \in A(s)$,  $\gammaleft \notin \Gamma$ and $\gammaright \notin \Gamma$}\right.}\\[1ex]
\text{where
  $\begin{aligned}[t]
    \gammaleft\ &=  x_{\pair{m}{s}}^\alpha \doteq 0 \text{ and }
    \gammaright\ =  x_{\pair{m}{s}}^\alpha > 0
  \end{aligned}$
  }
\end{gather*}}\ignorespaces
With a minor modification we get a calculus for \emph{deterministic} policies. It
only requires to re-define $\gammaright$ as $\gammaright =\ x_{\pair{m}{s}}^\alpha \doteq 1$.
As a benefit the program $\Gamma_\final$ will be linear.

After the \IR{A} rule has been applied exhaustively, for each $\alpha \in A(s)$ either
$x_{\pair{m}{s}}^\alpha > 0 \in \Gamma$ or $x_{\pair{m}{s}}^\alpha \doteq 0 \in \Gamma$.  If
$x_{\pair{m}{s}}^\alpha > 0 \in \Gamma$ we say that \define{$\alpha$ is prescribed in $\pair{m}{s}$
by $\Gamma$} and define \define{$\Prescribed(\pair{m}{s},\Gamma) = \{ \alpha \mid
x_{\pair{m}{s}}^\alpha > 0 \in \Gamma \}$}.

The set of prescribed actions in a policy state determines the
$\Succ$-relation  of the Markov chain under construction.  To get the required distribution over
enabled actions, it suffices to enforce a distribution over prescribed actions, with
this inference rule:
{
\begin{gather*}
 \infrule[$\Prescribed$]{
\Gamma \vdash \pair{m}{s}: \Psi
}{ 
\Gamma,  \gamma_{\pair m s}^\alpha \vdash \pair{m}{s}: \Psi
} \infrulecond{\left\{\pbox{ if $\Gamma \vdash \pair{m}{s}:\Psi$ is poised, \\
      $\alpha \in A(s)$ and $\gamma_{\pair m s}^\alpha \notin \Gamma$}\right.}\\[1ex]
\text{where   $\gamma_{\pair m s}^\alpha = \Sigma_{\alpha \in \Prescribed(\pair{m}{s},\Gamma)}\, x_{\pair{m}{s}}^\alpha \doteq 1$}
\end{gather*}}\ignorespaces
If the $\CHOOSE$ operator in step two selects the leftmost branch among the
\IR{A}-inferences then $\Gamma_\final$ contains $x_{\pair{m}{s}}^\alpha \doteq 0$, for all $\alpha \in A(s)$.
This is inconsistent with the constraint introduced by the
\IR{$\Prescribed$}-inference, corresponding to the fact that runs containing  $\pair
m s$ in this case do not  exist.
\subsection*{Blocking}
We are now turning to a ``loop check'' which is essential for termination, by, essentially,
blocking the expansion of certain states into successor states that do not mark progress.
For that, we need some more concepts.  For two nodes $u$ and $v$ in a branch we say that \define{$u$ is an ancestor of
$v$} and write $u \le v$ if $u = v$ or $u$ is closer to the root than $v$. An ancestor is
\define{proper}, written as $u < v$,  if $u \le v$ but $u \neq v$.
We say that two sequents $\Gamma_1\vdash\calF_1$ and $\Gamma_2\vdash\calF_2$ are
\define{indistinguishable} iff
$\calF_1 = \calF_2$, i.e., they differ only
in their $\Gamma$-components. Two nodes $u$ and $v$ are \define{indistinguishable} iff
their sequents are. We write $\Psi_u$ to denote the formula component of $u$'s label,
i.e., to say that the label is of the form $\Gamma \vdash \pair m s: \Psi_u$; similarly for
$\calF_u$ to denote $u$'s labelled formula.



\begin{definition}[Blocking]
  Let $w$ be a poised leaf and $v < w$ an ancestor node.
  If (i) $v$ and $w$ are indistinguishable, and (ii) for every $\TO X$-eventuality
  $\X (\psi_1 \U \psi_2)$ in $\Psi_v$ there is a node $x$ with $v < x \le w$ such that
  $\psi_2 \in \Psi_x$ then $w$ is \define{yes-blocked by $v$}.
    If there is an ancestor node $u < v$ such that
    (i) $u$ is indistinguishable from $v$ and $v$ is indistinguishable from
      $w$ (and hence $u$ is indistinguishable from $w$), and
      (ii) for every $\TO X$-eventuality $\X (\psi_1 \U \psi_2)$ in $\Psi_u$, if there is a
      node $x$ with $\psi_2 \in \Psi_x$ and $v < x \le w$ then there is a node $y$ with
      $\psi_2 \in \Psi_y$ and $u < y \le v$,
    then $w$ is \define{no-blocked by $u$}. 

    When we say that a sequent is yes/no-blocked we mean that its node is yes/no-blocked.
  \end{definition}
In the yes-blocking case all $\TO X$-eventualities in $\Psi_v$
become satisfied along the way from $v$ to $w$. This is why $w$ represents a success
case.
In the no-blocking case some $\TO X$-eventualities in $\Psi_v$ may have been satisfied
along the way from $u$ to $v$, but  not all, as this would be a yes-blocking
instead. Moreover, no progress has been made along the way from $v$ to $w$ for
satisfying the missing $\TO X$-eventualities. This is why $w$ represents a failure case.
The blocking scheme is adapted from~\cite{Reynolds:LTL-Tableaux:GandALF:2016} for LTL
satisfiability to our probabilistic case.
See~\cite{Reynolds:LTLTableau:long:2016,Reynolds:LTL-Tableaux:GandALF:2016} for
more explanations and examples, which are instructive also for its usage in our
framework.

Blocking is used in the following inference rules, collectively called the \IR{Loop}
rules. The node $v$ there is an
ancestor node of the leaf the rule is applied to.
{
\begin{gather*}
 \infrule[Yes-Loop]{
   \Gamma \vdash \pair{m}{s}:\Psi
}{ 
  \Gamma, x_{\pair m s}^\Psi \doteq (x_v)_{\pair m s}^\Psi  \vdash \textsf{Yes-Loop}
} \infrulecond{\text{ if $\Gamma \vdash \pair{m}{s}:\Psi$ is yes-blocked by $v$}}\\[-1mm]
 \infrule[No-Loop]{
  \Gamma \vdash \pair{m}{s}:\Psi
}{ 
  \Gamma, x_{\pair m s}^\Psi \doteq (x_v)_{\pair m s}^\Psi \vdash \textsf{No-Loop}
} \infrulecond{\text{ if $\Gamma \vdash \pair{m}{s}:\Psi$ is no-blocked by $v$}}
\end{gather*}}\ignorespaces
In either case, if $v$ is indistinguishable from $w$ then the probability of $\calF_v$
and $\calF_w$ are exactly the same, just because $\calF_v = \calF_w$.
This justifies adding $x_{\pair m s}^\Psi \doteq (x_v)_{\pair m s}^\Psi$.


The \IR{Loop} rules have a side-effect that we do not formalize: they add a link from
the conclusion node (the new leaf node) to the blocking node $v$, called the
\define{backlink}.  It turns the tableau into a graph that is no longer a tree.
The backlinks are used only for reachability analysis in step three of
the algorithm.
Figure~\ref{fig:bssc} has a graphical depiction. 

By preference of inference rules, the \IR{$\TO X$} rule introduced next can be
applied only if a \IR{Loop} rule does not apply.  The \IR{Loop} rules are at the core
of the termination argument.

This argument is standard for calculi based on
  formula expansion, as embodied in the \IR{$\TO U$} and \IR{$\neg\TO U$} rules: the sets
  of formulas obtainable by these rules is a subset of an a priori determined
  \emph{finite} set of formulas. This set consists of all subformulas of the given
  formula closed under negation and other operators. Any infinite branch hence would
  have to repeat one of these sets infinitely often, which is impossible with the
  \IR{loop} rules. Moreover, the state set $S$ and
  the mode set $M$ are finite and so the other rules do not cause problems either.  

For economy of notation, 
when $\Psi = \{\psi_1,\ldots,\psi_n\}$, for some $\psi_1,\ldots,\psi_n$ and $n > 0$, let
\define{$\X \Psi$} denote the set $\{\X \psi_1,\ldots, \X \psi_n\}$.
{
\begin{gather*}
 \infrule[$\TO X$]{
    \Gamma \vdash \pair{m}{s}: \X \Psi
}{ 
  \Gamma \vdash \pair{m'}{t_1}: \Psi \quad \cup \quad \cdots \quad \cup \quad \Gamma \vdash
  \pair{m'}{t_{k-1}}: \Psi  \quad \cup \quad \Gamma,\gamma_1 \vdash \pair{m'}{t_k}: \Psi 
}\\[1mm]
\text{where}
\begin{aligned}[t]
  m' & \ = \ \Delta(m,s) \\
  \{ t_1,...,t_k \} &\ =\ \textstyle\bigcup_{\alpha \in
  \Prescribed(\pair{m}{s},\Gamma)}\Succ(s,\alpha) \text{ ,  for some $k \ge 0$}\\
  \gamma_1 &\ =\ x_{\pair{m}{s}}^{\X \Psi} \doteq
       \Sigma_{\alpha \in \Prescribed(\pair{m}{s},\Gamma)}\, [x_{\pair{m}{s}}^\alpha \cdot
          (\Sigma_{t \in \Succ(s, \alpha)}\, P(t|s,\alpha) \cdot
           x_{\pair{m'}{t}}^{\Psi})]
      \end{aligned}
\end{gather*}}\ignorespaces
This is the (only) rule for expansion into successor states.

If $u$ is the node the \IR{$\TO X$} rule is applied to and $u_1,\ldots,u_k$ are its
children then each $u_i$ is called an \define{\IR{$\TO X$}-successor (of $u$)}.

The \IR{$\TO X$} rule follows the set of actions prescribed in $\pair{m}{s}$ by
$\Gamma$ through to successor states. This requires summing up the probabilities of
carrying out $\alpha$, as represented by $x_{\pair{m}{s}}^\alpha$, multiplied by the sums of
the successor probabilities weighted by the respective transition probabilities.
This is expressed in the constraint $\gamma_1$.
Only these $k$ successors need to be summed up, as all other, non-prescribed
successors, have probability 0.  

\subsection{Forcing Probabilities}
We are now turning to the $\FORCE$ operator which we left open in step three of the algorithm.
\begin{wrapfigure}[14]{r}{0.4\textwidth}
  \vspace*{-1.5ex}
  \includegraphics[width=0.4\textwidth]{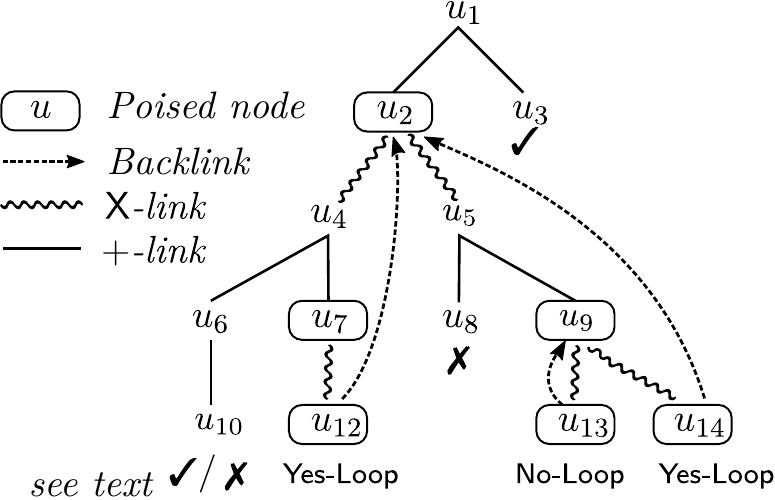}
   \caption{An example tableau $T$ from step 2. The subgraph below $u_2$ is a
     strongly connected component if $u_{10}$ is $\Cross$-ed.}
   \label{fig:bssc}
\end{wrapfigure}
It forces a probability 0 or 1 for certain labelled formulas occurring in a
bottom strongly connected component in a tree
from step two.
The tree %
in the figure to the right helps to illustrate the
concepts introduced in the following.
%

We need some basic notions from graph theory.
A subset $M$ of the nodes $N$ of a given graph is \define{strongly connected} if, for
each pair of 
nodes $u$ and $v$ in $M$, $v$ is reachable from $u$ passing only through states in $M$.
A \define{strongly connected component (SCC)} is a maximally strongly connected set
of nodes (i.e., no superset of it is also strongly connected).
A \define{bottom strongly connected component (BSCC)} is a SCC $M$ from which no
state outside $M$ is reachable from $M$.

Let $T = \CHOOSE(\TABLEAU(\Gamma \vdash \calF))$ be a tree without don't-know branching 
obtained in step~2. 
We wish to take $T$ together with its backlinks as the graph under
consideration and analyse its BSCCs. However, for doing so we cannot take $T$ as it
is. On the one hand, our tableaux describe state transitions introduced by
\IR{$\TO X$} rule applications. Intuitively, these are amenable to BSCC analysis as
one would do for state transition systems. On the other hand, $T$ has interspersed rule
applications for analysing Boolean structure, which distort the state transition
structure. These rule applications have to be taken into account prior to the BSCC
analysis proper. 

For this, we distinguish between $\TO X$-links and $+$-links in $T$. An
$\TO X$-link is an edge between a node and its child if the \IR{$\TO X$} rule
was applied to the node, making its child an $\TO X$-successor, otherwise it is a
$+$-link. (``$+$-link'' because probabilities are summed up.)

Let $u$ be a node in $T$ and \define{$\Subtree_T(u)$}, or just
\define{$\Subtree(u)$}, the subtree of $T$ rooted at $u$ without the backlinks.
We say that $u$ is a \define{0-deadend (in $T$)} if $\Subtree_T(u)$ has no
$\TO X$-links and every leaf in $\Subtree_T(u)$ is $\Cross$-ed.  In a
0-deadend the probabilities all add up to a zero probability for the pivot of $u$. This is
shown by an easy inductive argument. 

\begin{definition}[Ambiguous node]
  \label{def:unambiguous}
  Let $u$ be a node in $T$. We say that \define{$u$ is ambiguous (in $T$)} iff
  (i) $\Subtree_T(u)$ contains no $\Checkmark$-ed leaf, and
  (ii) $\Subtree_T(u)$ contains no $\TO X$-successor 0-deadend node.
We say that \define{$u$ is unambiguous} iff $u$ is not ambiguous.
\end{definition}
The main application of Definition~\ref{def:unambiguous} is when the node $u$ is the
root of a BSCCs, defined below. The probability of $u$'s pivot $\pair{m}{s}:\Psi$ then is not uniquely
determined. This is because expanding $u$
always leads to a cycle, a node with the same pivot, and there is no escape from that
according to conditions (i) or (ii) in Definition~\ref{def:unambiguous}.
In other words, the probability of $\pair{m}{s}:\Psi$ is defined only in terms of
itself.\footnote{In terms of the resulting program, 
  $(x_u)_{\pair m s}^\psi$ is not constrained to any specific value in $[0..1]$.
  This can be shown by ``substituting in'' the equalities in $\Gamma_\final$ for the probabilities
of the pivots in the subtrees  below $u$ and arithmetic simplifications. }

In the figure above, the node $u_1$ is unambiguous because of case (i) in
Definition~\ref{def:unambiguous}. Assuming $u_{10}$ is
$\Checkmark$-ed, the node $u_2$ is unambiguous by case (i). The pivot in
$u_{10}$, then, has probability 1 which is propagated upwards to $u_4$ (and enforces
probability 0 for the pivot of $u_7$). It contributes a non-zero probability to the
transition from $u_2$ to $u_4$ and this way escapes a cycle.
If  $u_{10}$ is $\Cross$-ed, the node $u_2$
is ambiguous.

If case (ii) in Definition~\ref{def:unambiguous} is violated there is an $\TO X$-successor
node whose pivot has probability 0. Because every $\TO X$-link has a non-zero transition
probability, the probabilities obtained through the other $\TO X$-successor nodes add
up to a value strictly less than 1. This also escapes the cycle leading to underspecified programs
(not illustrated above).

Let \define{$0(T)\!=\!\{ w\!\mid\!\text{$w$ is a node in some 0-deadend of $T$} \}$} be
all nodes in all 0-deadends in $T$. In the example,
$0(T)\!=\! \{u_6, u_{10}, u_8 \}$ if $u_{10}$ is $\Cross$-ed and
\mbox{$0(T)\!=\!\{u_8 \}$ if $u_{10}$ is $\Checkmark$-ed.}

Let $u$ be a node in $T$ and \define{$M(u) = \{ w \mid \text{$w$ is a node in $\Subtree(u)$} \}
\setminus 0(T)$}.  That is, $M(u)$ consists of the nodes in the subtree rooted at $u$ after
ignoring the nodes from the 0-deadend subtrees.
In the example
$M(u_2) = \{u_2, u_4, u_5, u_7, u_9, u_{12}, u_{13}, u_{14}\}$ if $u_{10}$ is $\Cross$-ed. 
If $u_{10}$ is $\Checkmark$-ed then $u_6$ and $u_{10}$ have to be added.

We say that \define{$u$ is the root of a BSCC (in $T$)} iff $u$ is poised, ambiguous and
$M(u)$ is a BSCC in $T$ (together with the backlinks).
%
In the example, assume
%
%
that $u_{10}$ is $\Cross$-ed. Then $u_2$ is
poised, ambiguous and the root of a BSCC. In the example, that $M(u_2)$ is a BSCC in
%
%
$T$ is easy to verify.

Now suppose that $u$ is the root of a BSCC with pivot  $\pair{m}{s}: \X \Psi$. This means that the
probability of $\pair{m}{s}: \X \Psi$ is not uniquely determined. This situation then is
fixed by means of the $\FORCE$  operation, generally defined as follows: 
\begin{align*}
  \BSCC(T) & := \{ u \mid \text{$u$ is the root of a BSCC in $T$}\, \}\\[-1mm]
  \FORCE(T) & := \{ (x_u)_{\pair{m}{s}}^{\X \Psi} \doteq \chi  \mid \pbox[t]{ $u \in \BSCC(T)$, and\\
  if some leaf of the subtree rooted at $u$ is a \IR{Yes-Loop} \\
  \qquad then $\chi = 1$ else $\chi = 0$  $\}$}
\end{align*}
That is, $\FORCE$ing removes the ambiguity for the probability of the pivot
$\pair m s : {\X \Psi }$ at the root $u$ of a BSCC by setting it to 1 or to 0.
If $\FORCE$ing adds $(x_u)_{\pair{m}{s}}^{\X \Psi} \doteq 1$ then there is a run that satisfies every
$\TO X$-eventuality in $\X \Psi$, by following the branch to a \textsf{Yes-Loop}.
Because we are looking at a BSCC, for fairness reasons, \emph{every} run will do
this, and infinitely often so, this way giving $\X \Psi$ probability 1. Otherwise, if
there is no \textsf{Yes-Loop}, there is
some $\TO X$-eventuality in $\X \Psi$ that cannot be satisfied, forcing probability 
0.



\section{Conclusions and Future Work}
\label{sec:conclusions}
In this paper we presented a first-of-its kind algorithm for the controller synthesis problem for Markov
Decision Processes whose intended behavior is described by PCTL* formulas. The only
restriction we had to make -- to get decidability -- is to require policies with finite
history. We like to propose that the description of the algorithm 
is material enough for one paper, and so we leave many interesting questions for
future work.

The most pressing theoretical question concerns the exact worst-case
complexity of the algorithm. Related to that, it will be interesting
to specialize and analyze our framework for fragments of PCTL*, such
as probabilistic LTL and CTL or simpler fragments and
restricted classes of policies that might lead to \emph{linear}
programs (and ideally to solving only a polynomial number of such
programs).
For instance, we already mentioned that computing deterministic policies leads to linear programs
in our tableau (see the description of the \IR{\textsf{A}} inference rule how this is done.) 
Moreover, it is well-known that cost-optimal stochastic policies for classes of MDPs with simple constraints bounding the
probability of reaching a goal state can be synthesized in linear time in the
size of the MDP by solving a \emph{single} linear program \cite{altman:99,DBLP:conf/ijcai/DolgovD05}. An interesting
question is how far these simple constraints can be generalised towards PCTL* whilst
remaining in the linear programming framework (see e.g. \cite{DBLP:conf/aaai/SprauelKT14}).


On implementation,
a na{\"\i}ve implementation of the algorithm as presented above would perform
poorly in practice. However, it is easy to exploit some straightforward
observations for better performance. For instance, steps one (tableau construction)
and two (committing to a don't-know non-deterministic choice) should be combined into
one. Then, if a don't know non-deterministic inference rule is carried out the first
time, every subsequent inference with the same rule and pivot can be forced to the same
conclusion, at the time the rule is applied. Otherwise an inconsistent
program would result, which never needs to be searched for.
Regarding space, although all children of a union branching
inference rule need to be expanded, this does not imply they always all need to be kept in
memory simultaneously. Nodes can be expanded in a one-branch-at-a-time
fashion and using a global variable for $\Gamma_\final$ for collecting the programs in the
leaves of the branches \emph{if they do not belong to a bottom strongly connected
component}. Otherwise, the situation is less obvious and we leave it to future work.
Another good source of efficiency improvements
comes from more traditional tableau. It will be mandatory to
exploit techniques such as dependency-directed backtracking, lemma
learning, and early failure checking for search space pruning.

\subsection*{Acknowledgements}

This research was funded by AFOSR grant FA2386-15-1-4015. We would also like to thank
the anonymous reviewers for their constructive and helpful comments.

\section{Additional Operators and Useful Equivalences}
\label{sec:additional-operators}
Additional operators can be defined on top of the temporal operator $\TO U$, as usual. In particular,
$\opFalse := \neg\opTrue$ and
\begin{align*}
  \F \psi & := \opTrue \U \psi & \
   \G \psi & := \neg \F \neg  \psi \quad \text{(${} \equiv \opFalse \R \psi$)} \\
  \psi_1 \R \psi_2 & := \neg(\neg \psi_1 \U \neg \psi_2) &
  \psi_1 \W \psi_2 & := (\psi_1 \U \psi_2) \vee \G \psi_1 \quad \text{(${} \equiv \psi_2 \R (\psi_2 \vee \psi_1))$}
\end{align*}
For the ``release'' operator $\TO{R}$, the formula $\psi_1 \R \psi_2$ says
that $\psi_2$ remains true until and including once $\psi_1$ becomes true;
if $\psi_1$ never become true, $\psi_2$ must remain true forever. 
Regarding the ``weak until'' operator $\TO{W}$, the formula $\psi_1 \W \psi_2$ is similar to
$\psi_1 \U \psi_2$ but the stop condition $\psi_2$ is not required to occur. In that case
$\psi_1$ must remain true forever.

\define{Distributivity laws}:
\begin{align}
  \X (\psi_1  \lor  \psi_2 ) &\equiv  (\X \psi_1 ) \lor  (\X \psi_2 ) &
 \X (\psi_1  \land  \psi_2 )& \equiv  (\X \psi_1 ) \land  (\X \psi_2 ) \label{eq:d1} \tag{D1}\\
  \X (\psi_1  \U \psi_2 )& \equiv  (\X \psi_1 ) \U (\X \psi_2 )  &
\X (\psi_1  \R \psi_2 )& \equiv  (\X \psi_1 ) \R (\X \psi_2 ) \tag{D2} \\
  \F (\psi_1  \lor  \psi_2 ) &\equiv  (\F \psi_1 ) \lor  (\F \psi_2)&
 \G (\psi_1  \land  \psi_2 )&\equiv  (\G \psi_1 ) \land  (\G \psi_2 ) \tag{D3}\\
  \psi  \U (\psi_1  \lor  \psi_2 ) &\equiv  (\psi  \U \psi_1 ) \lor  (\psi  \U \psi_2 ) &
 (\psi_1  \land  \psi_2 ) \U \psi  &\equiv (\psi_1  \U \psi ) \land  (\psi_2  \U \psi ) \tag{D4}
\end{align}

\define{Negation propagation laws}:
\begin{align}
  \lnot \X \psi_1  &\equiv  \X \lnot \psi_1 & \lnot \G \psi_1  &\equiv  \F \lnot \psi_1 & \lnot \F \psi_1  &\equiv  \G \lnot
                                                                \psi_1  \label{eq:n1} \tag{N1}\\
\lnot  (\psi_1  \U \psi_2 ) &\equiv  (\lnot \psi_1  \R \lnot \psi_2) & \lnot  (\psi_1  \R \psi_2 ) &\equiv  (\lnot \psi_1  \U \lnot \psi_2 ) \label{eq:n2} \tag{N2}
\end{align}

\define{Absorption laws}:
\begin{align}
  \F \F \psi &\equiv \F \psi & \G \G \psi &\equiv \G \psi  \tag{A1} \\
  \F \G \F \psi &\equiv \G \F \psi & \G \F \G \psi &\equiv \F \G \psi  \tag{A2}
\end{align}

\define{Expansion laws}:
\begin{align}
  \psi_1 \U \psi_2 & \equiv \psi_2 \vee (\psi_1 \wedge \X (\psi_1 \U \psi_2 )) &
  \neg(\psi_1 \U \psi_2) & \equiv \neg\psi_2 \land (\neg \psi_1 \vee \X \neg(\psi_1 \U \psi_2 ))  \tag{E}
\end{align}

For ${\sim} \in \{ <, \leq, >, \geq \}$ define the operators $\compl{\sim}$ and $[\sim]$ as follows:
\begin{align*}
{\compl{<}} & ={\ge} &
{\compl{\leq}} & ={>} &
{\compl{>}} & ={\leq} &
{\compl{\ge}} & ={<} \\
[<] & ={>} &
[\leq] & ={\ge} &
[>] & ={<} &
[\ge] & ={\le} 
\end{align*}

Some of the following equivalences cannot be used for ``model checking'' PCTL* (the
left~(\ref{eq:p1a}) equivalence, to be specific) where actions are implicitly
universally quantified. In the context of Markov Chains, which we implicitly have,
there is no problem:
\begin{align}
  \neg \P_{\sim z} \psi & \equiv \P_{\compl{\sim} z} \psi &
  \P_{\sim z} \neg \psi & \equiv \P_{\mathop{[\sim]} 1-z} \psi \label{eq:p1a} \tag{P1} \\
  \P_{\ge 0} \psi & \equiv \opTrue & \P_{> 1} \psi & \equiv \opFalse \label{eq:p2a} \tag{P2}\\
  \P_{\le 1} \psi  & \equiv \opTrue & \P_{< 0} \psi & \equiv \opFalse \label{eq:p3a} \tag{P3}\\
  \P_{\geq u} \P_{\sim z} \psi & \equiv \P_{\sim z} \psi \quad\text{if $u \neq 0$} &
  \P_{> u} \P_{\sim z} \psi & \equiv \P_{\sim z} \psi \quad\text{if $u\neq 1$} \label{eq:p4a} \tag{P4}\\
  \P_{\leq u} \P_{\sim z} \psi & \equiv \P_{\ge 1-u}\P_{\compl{\sim} z} \psi &
  \P_{< u} \P_{\sim z} \psi & \equiv \P_{> 1-u}\P_{\compl{\sim} z} \psi \label{eq:p5a} \tag{P5}
\end{align}
Some notes on these equivalences.
The left equivalence of~(\ref{eq:p1a}) is trivial and the right equivalence
uses the fact that $\Pr^{\calM_\pi}\{s\cdots  \mid \calM_\pi, s\cdots \models \psi\} +
\Pr^{\calM_\pi}\{s\cdots  \mid \calM_ \pi, s\cdots \models \neg\psi\} = 1$.
The equivalences~(\ref{eq:p2a}) and (\ref{eq:p3a}) are again trivial.
For the equivalences~(\ref{eq:p4a}) it helps to observe that the subformula
$\P_{\sim z} \psi$ is both a state and a path formula. As a state formula it evaluates
in a given context $\calM,\pi,s$ to either true or false. Taken as a path formula of the
outer $\PQ P$ quantifier it hence stands for the set of either \emph{all} runs from
$s$ or the empty set, respectively. With this observation the
equivalences~(\ref{eq:p4a}) follow easily from the semantics of the $\PQ P$ operator.
The equivalences~(\ref{eq:p5a}) are obtained by first applying the right
equivalence in~(\ref{eq:p1a}) (from right to left) and then the left equivalence
in~(\ref{eq:p1a}). 

\section{Example}
\label{sec:example}
Consider the MDP in Figure~\ref{fig:mdp-1} and the partially specified finite-memory
policy $\pi_\fin = (\{\m\}, \start, \Delta, \cdot)$ with a single mode $\m$, making
$\pi_\fin$ Markovian.  The functions $\start$ and $\Delta$ hence always return $\m$,
allowing us to abbreviate $\pair m {s_i}$ as just $s_i$.
Action $\beta$ leads non-deterministically to states $s_2$ and $s_3$, each with
probability 0.5. The actions $\alpha_i$ for $i \in \{1,2,3\}$ are self-loops with
probability one (not shown).  The label set of $s_2$ is $\{a\}$ in all other states
it is empty.  

The example is admittedly simple and is only from the PCTL subset
of CTL*. But it suffices to show the main aspects of the calculus.
\begin{figure}[htp]
  \centering
  \includegraphics[scale=0.9]{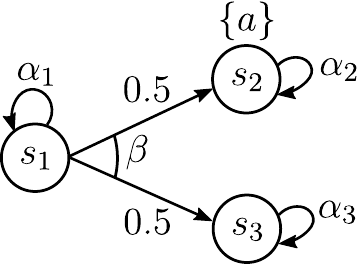}
  \caption{The transitions of the example MDP $\calM$ depicted as a graph. The
    initial state is $s_1$. Action $\beta$ leads non-deterministically to 
    states $s_2$ and $s_3$, each with probability 0.5. The actions $\alpha_i$ for $i \in
    \{1,2,3\}$ are self-loops with probability one (not shown).
    The label set of $s_2$ is $\{a\}$ in all other states it is empty. 
  }
  \label{fig:mdp-1}
\end{figure}

Let the state formula of interest be $\phi = \P_{\ge 0.3} \F \G a$. We
wish to obtain a $\Gamma_\final$ such that any solution $\sigma$ synthesizes a suitable
$\act_\sigma$, i.e., the policy $\pi_\fin(\sigma)$ completed by $\sigma$ satisfies 
$\calM_{\pi_\fin(\sigma)}, s_1 \models \phi$.

The BSCCs depend on whether $\act_\sigma(\m, s_1, \beta) > 0$ holds, i.e., if
$\beta$ can be executed at $s_1$. 
(This is why the calculus needs to make a corresponding guess, with its \IR{A}-rule.)
If not,
then $s_2$ and $s_3$ are unreachable, and the self-loop at $s_1$ is the only BSCC,
which does not satisfy $\G a$. 
If yes, then there are two BSCCs, the
self-loop at $s_2$ and the self-loop at $s_3$, and the BSCC at $s_2$ satisfies $\G a$. 
By fairness of execution, with probability one some BSCC will be reached, and the BSCC
at $s_2$ is reached with probability 0.5, hence,  if $\act_\sigma(\m, s_1, \beta) > 0$.
In other words, devising \emph{any} policy that reaches $s_2$ will hence suffice to
satisfy $\phi$.  The expected result thus is just a constraint on $\sigma$ saying
$\act_\sigma(\m, s_1, \beta) > 0$ and the derivation will indeed show that by
$\CHOOSE$ing a branch with $x_{s_1}^\beta > 0$ in $\Gamma_\final$.

Figures~\ref{fig:deriv-1} to \ref{fig:deriv-5} summarize the derivation from the
initial sequent $\sigma_1 = x_{s_1}^{\{\phi\}} \doteq 1 \vdash s_1:\phi$.  In these figures we
write, for brevity, $\Gamma \vdash \psi, \Psi$ instead of $\Gamma \vdash \{\psi\} \uplus \Psi$. For better readability we write
$\G \psi$ as a macro for $\neg \F \neg \psi$ and $\F \psi$ for $\opTrue \U \psi$. Please consider the
captions in the figures for further explanations on notation.

Tree \circled{1} in Figure~\ref{fig:deriv-1} shows the derivation of
$\CHOOSE(\TABLEAU(\sigma_1))$. It has only one branch which ends in the leaf node $u_\circled{1}$.
Because there are no BSCCs in tree \circled{1}, $\FORCE$ does not add constraints, and
therefore $\Gamma_\final = \Gamma_{u_\text{\circled{1}}}$ (the constraint system of $u_\text{\circled{1}}$). Notice that if the derivation had
$\CHOOSE$n the right branch at the top of tree \circled{1}, $\Gamma_\final$ would be
unsatisfiable (as indicated in Figure~\ref{fig:deriv-1}).  This case is
uninteresting and we do not consider it further.
\begin{figure}[htp]
  \centering
  \includegraphics[scale=0.9]{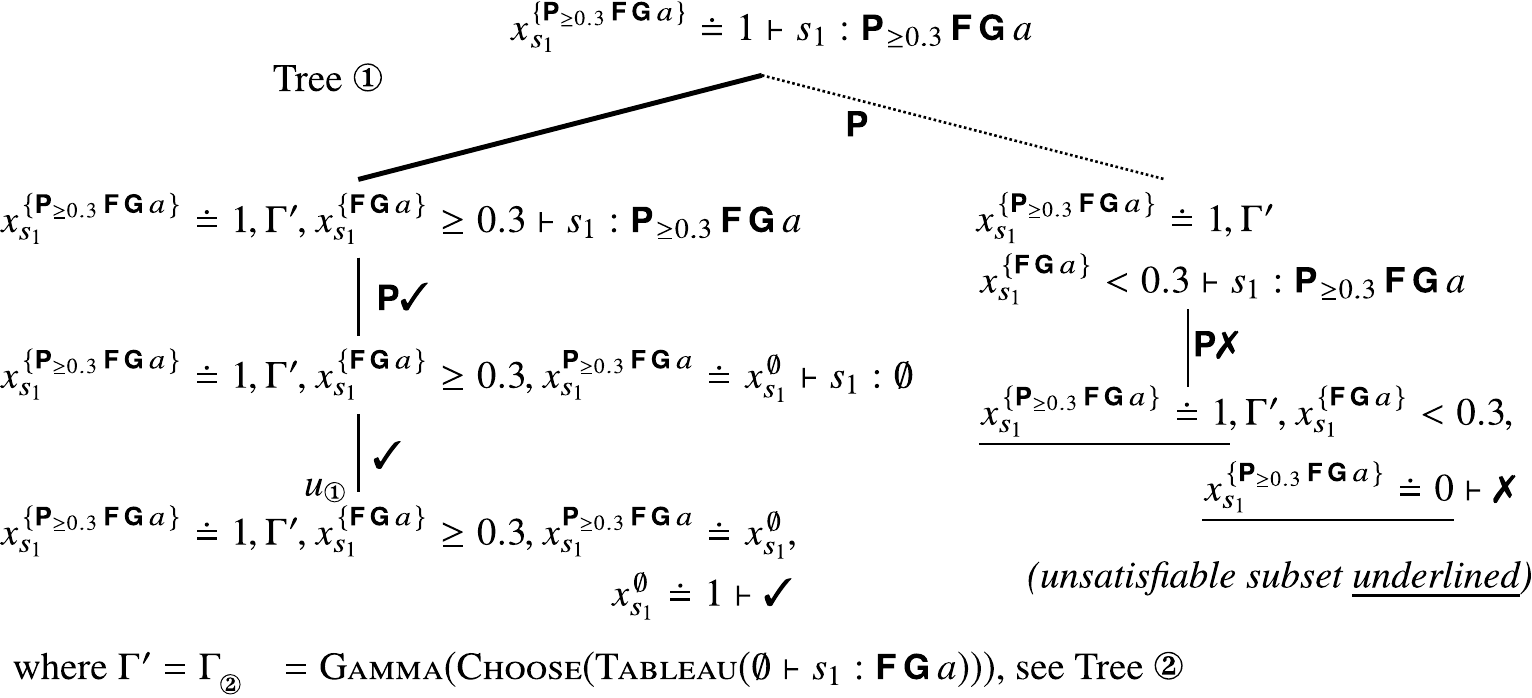}
  \caption{Start of derivation for synthesizing a policy for the MDP in
    Figure~\ref{fig:mdp-1} for the formula $\P_{\ge 0.3} \F \G a$.
    The links are annotated with the name of the inference rule applied.
    The $\CHOOSE$n
    alternative in the \IR{$\PQ P$} inference is the bold link, the dotted link is the
    non-$\CHOOSE$n alternative. The leaf node of
    the left branch is called $u_{\text{\circled{1}}}$.  }
  \label{fig:deriv-1}
\end{figure}

Figures~\ref{fig:deriv-2}-\ref{fig:deriv-5} in combination show the derivation from initial sequent
$\emptyset \vdash s_1: \F \G a$ for computing $\Gamma_\circled{2}$ in Figure~\ref{fig:deriv-1}. The dotted
line in Figure~\ref{fig:deriv-2} represents the contribution of tree \circled{3} to
tree \circled{2} in terms of probabilities for the pivot $s_1: \F \G a$ in the root of tree
\circled{2}. More precisely, the constraint $\Gamma_\circled{2}$ will
enforce $x_{s_1}^{\{\F\G\  a\}} = x_{s_1}^{\{\X\F\G\ a\}}$, where
$x_{s_1}^{\{\X\F\G\  a\}}$ holds the probability of $s_1: \X\F\G\ a$ as per tree
\circled{3}. This equality holds because all branches in tree \circled{2} are
\Cross-ed, therefore each contributing a value 0 to the sums in the union
branches. (Figures~\ref{fig:deriv-4} and ~\ref{fig:deriv-5} have more such examples.)
\begin{figure}[htp]
  \centering
  \includegraphics[scale=0.9]{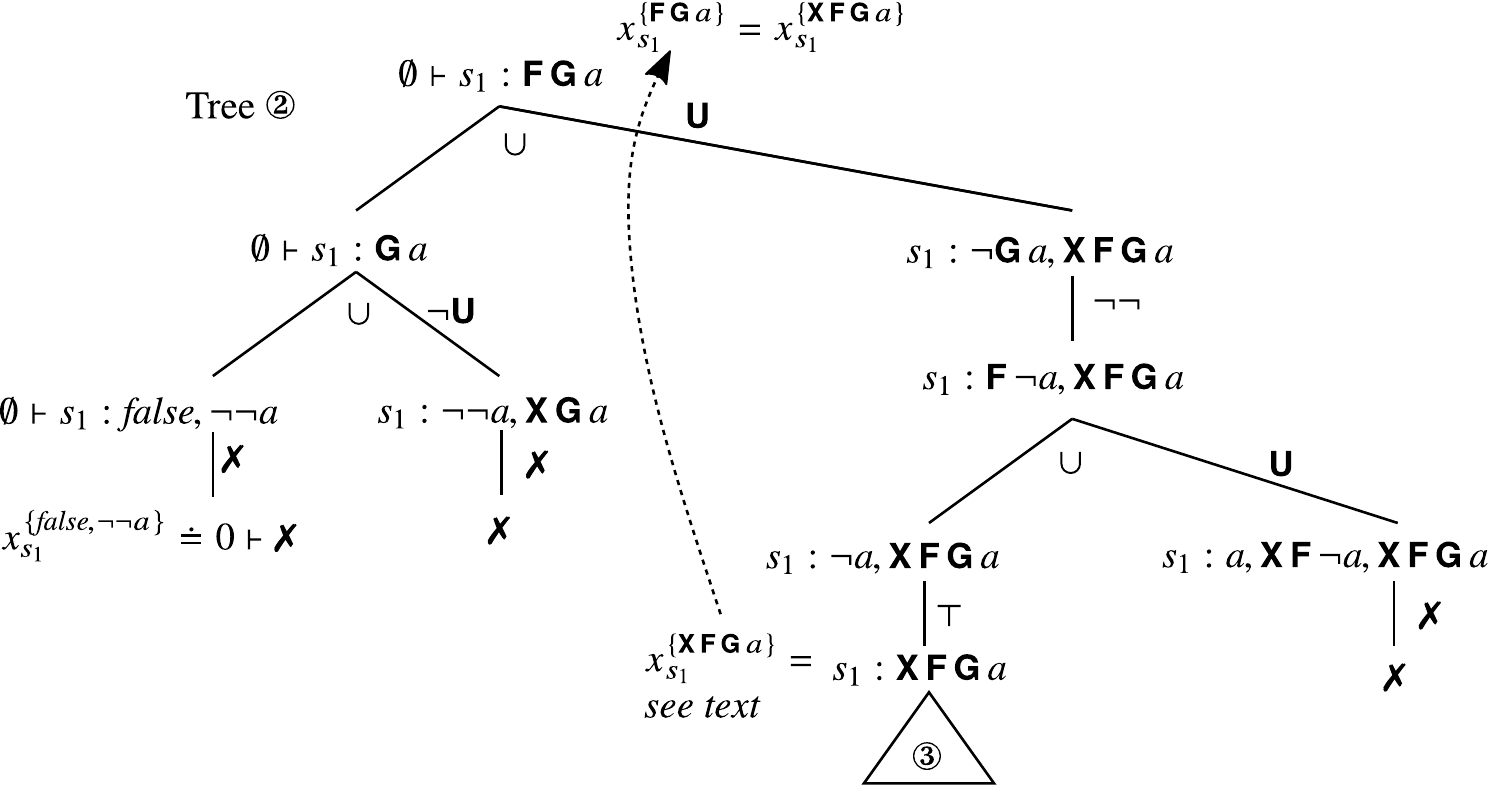}
  \caption{First part of a derivation from $\emptyset \vdash s_1: \F \G a$ for computing $\Gamma_\circled{2}$ in Figure~\ref{fig:deriv-1}. For space
    reasons we only write the right hand sides of sequences or only write interesting
    parts of the left hand sides. The dotted line represents a dependency of
    probability constraints. The derivation for tree \circled{3} is in Figure~\ref{fig:deriv-3}.}
  \label{fig:deriv-2}
\end{figure}

Figure~\ref{fig:deriv-3} has the tree \circled{3} with pivot $s_1: \X \F \G a$ at its
root. Again the $\CHOOSE$n alternatives are already highlighted, this time for the
\IR{\textsf{A}} inferences. ($\CHOOSE$ing the $x_{s_1}^{\alpha_1} \doteq 0$ branch would also lead to
a policy, but obviously $\beta$ must be prescribed in $s_1$ for being able to reach $s_2$.)
\begin{figure}[htp]
  \centering
  \includegraphics[scale=0.9]{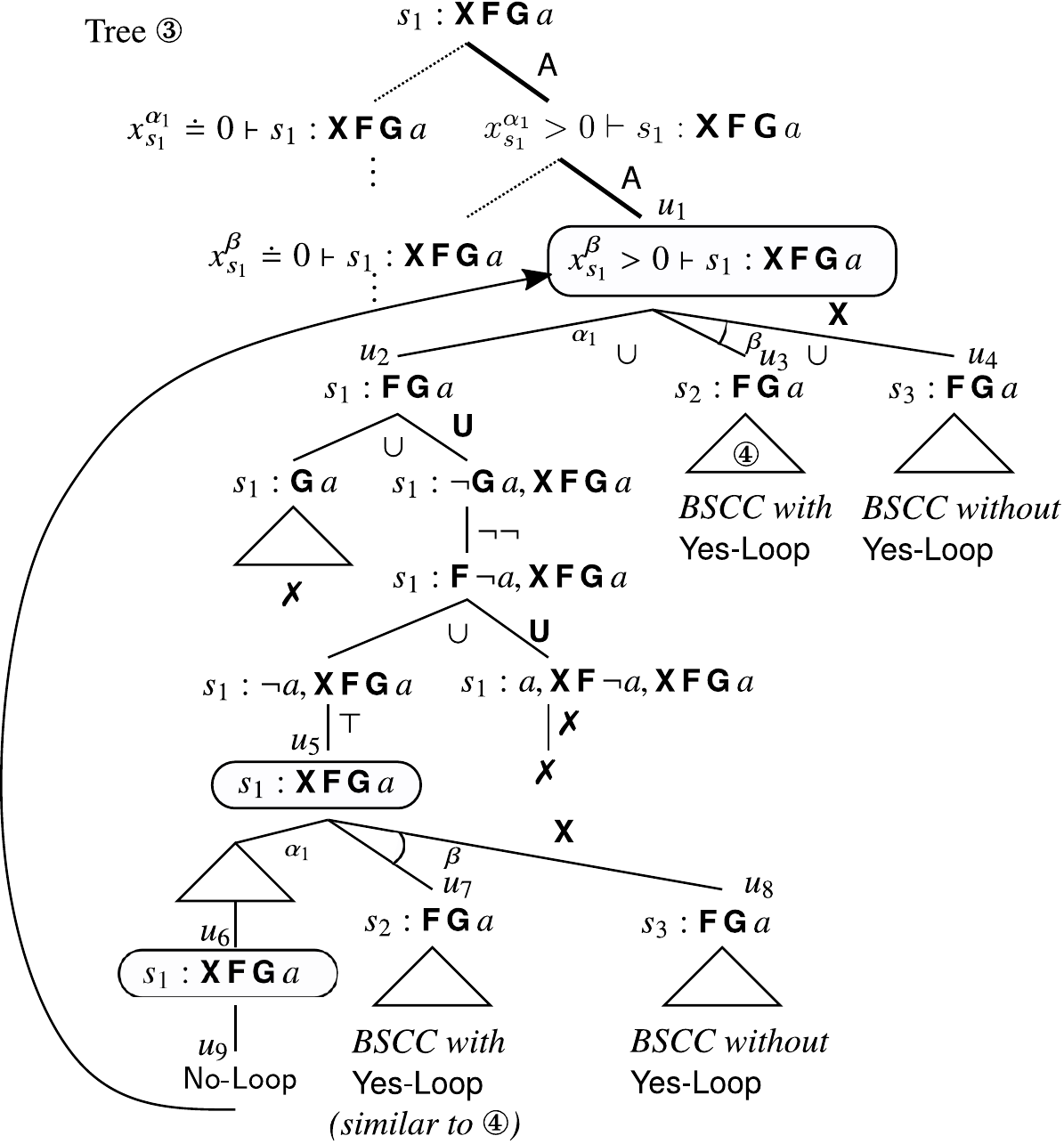}
  \caption{Sub-derivation with Tree \circled{3} continuing
    Figure~\ref{fig:deriv-2}. Poised nodes are framed.
    The links in the \IR{$\TO X$} inference are annotated
    with actions. An triangle with a \Cross is
    a tree all whose leaves are \Cross-ed. A triangle with links at the bottom is a
    tree with branches into the linked nodes and all whose leaves are \Cross-ed.  A triangle labelled
    ``\textit{BSCC with/without a \textnormal{\textsf{Yes-Loop}}}'' is
      a tree that contains a BSCC with/without \textsf{Yes-Loop} and all non-loop
      leaves are \Cross-ed. Alternative choices for the \IR{A} rule are only indicated
      as dots, as they are not relevant in this example. See text for propagation of constraints.}
  \label{fig:deriv-3}
\end{figure}

It is instructive to see how tree \circled{3} contributes to $\Gamma_\circled{2}$ a constraint
for $\smash{x_{s_1}^{\{\X\F\G a\}}}$, i.e., the pivot at its root. For that, assume that
the trees annotated as \textit{BSCC with (without) \textnormal{\textsf{Yes-Loop}}} all
contribute a probability of one (zero, respectively) to $\Gamma_\circled{2}$. (We spell this out
in detail only in one case, cf.\ tree \circled{4} in Figure~\ref{fig:deriv-4}.)
Combining all constraints in the leaves of tree \circled{3} then entails the
following set of equalities (the comment $u_1 \leftarrow u_2, u_3, u_4$ indicates the
dependency of the left hand side of its equality from the right hand side in terms of
corresponding nodes):
\begin{align}
  (x_{u_1})_{s_1}^{\{\X\F\G a\}} & =
                           x_{s_1}^{\alpha _1}\cdot (x_{u_2})_{s_1}^{\{\F\G a\}} + 
                           x_{s_1}^{\beta}\cdot 0.5 \cdot (x_{u_3})_{s_2}^{\{\F\G a\}} + 
                           x_{s_1}^{\beta}\cdot 0.5 \cdot (x_{u_4})_{s_3}^{\{\F\G a\}}
                           \tag{$u_1 \leftarrow u_2, u_3, u_4$} \\ 
  (x_{u_2})_{s_1}^{\{\F\G a\}} &= 
                             (x_{u_5})_{s_1}^{\{\X\F\G a\}} \tag{$u_2 \leftarrow u_5$}\\ 
  (x_{u_3})_{s_2}^{\{\F\G a\}} & = 1 \tag{$u_3$}\\ 
  (x_{u_4})_{s_3}^{\{\F\G a\}} &= 0 \tag{$u_4$} \\ 
(x_{u_5})_{s_1}^{\{\X\F\G a\}} &= 
                           x_{s_1}^{\alpha _1}\cdot (x_{u_6})_{s_1}^{\{\X\F\G a\}} + 
                           x_{s_1}^{\beta}\cdot 0.5 \cdot (x_{u_7})_{s_2}^{\{\F\G a\}} + 
                             x_{s_1}^{\beta}\cdot 0.5 \cdot (x_{u_8})_{s_3}^{\{\F\G a\}}
  \tag{$u_5 \leftarrow u_6, u_7, u_8$} \\ 
  (x_{u_6})_{s_1}^{\{\X\F\G a\}} & = (x_{u_9})_{s_1}^{\{\X\F\G a\}} \tag{$u_6 \leftarrow u_9$} \\
  (x_{u_7})_{s_2}^{\{\F\G a\}} & = 1 \tag{$u_7$} \\
  (x_{u_8})_{s_3}^{\{\F\G a\}} &= 0 \tag{$u_8$}\\
  (x_{u_9})_{s_1}^{\{\X \F\G a\}} &= (x_{u_1})_{s_1}^{\{\X \F\G a\}} \tag{$u_9 \leftarrow u_1$}
\end{align}
Substituting in yields a simplified set of equalities:
\begin{align}
(x_{u_1})_{s_1}^{\{\X\F\G a\}} & =
                           x_{s_1}^{\alpha _1}\cdot (x_{u_5})_{s_1}^{\{\X\F\G a\}} + 
                           0.5 \cdot x_{s_1}^{\beta} \tag{$u_1 \leftarrow u_2, u_3, u_4$} \\ 
(x_{u_5})_{s_1}^{\{\X\F\G a\}} &= 
                           x_{s_1}^{\alpha _1}\cdot (x_{u_1})_{s_1}^{\{\X\F\G a\}} + 
                           0.5 \cdot x_{s_1}^{\beta}  \tag{$u_5 \leftarrow u_6, u_7, u_8$} 
\end{align}
%
%
Isolating $(x_{u_1})_{s_1}^{\{\X\F\G a\}}$ on the left hand side:
\begin{equation*}
  (x_{u_1})_{s_1}^{\{\X\F\G a\}}\cdot(1 - (x_{s_1}^{\alpha _1})^2) =
                           0.5 \cdot x_{s_1}^{\beta}\cdot(1 + x_{s_1}^{\alpha _1})  
                         \end{equation*}
Both $\alpha_1$ and $\beta$ are prescribed actions thanks to the two
\IR{\textsf{A}} inferences preceding the node $u_1$, resulting in
$\smash{\{x_{s_1}^{\alpha_1} > 0, x_{s_1}^{\beta} > 0\} \subset \Gamma_\circled{2}}$.
The \IR{$\TO X$} inference
at $u_1$ adds the constraint $x_{s_1}^{\alpha_1} + x_{s_1}^{\beta} \doteq 1$ to
$\Gamma_\circled{2}$ (cf.\ $\gamma_2$ in the definition of the \IR{$\TO X$} rule).
It follows $x_{s_1}^{\beta} = 1 - x_{s_1}^{\alpha_1}$. Substituting into the previous
equality we get
\begin{equation*}
  (x_{u_1})_{s_1}^{\{\X\F\G a\}}\cdot(1 - (x_{s_1}^{\alpha _1})^2)  =
                           0.5 \cdot (1 - x_{s_1}^{\alpha_1})\cdot(1 + x_{s_1}^{\alpha _1}) 
  =                            0.5 \cdot (1 - (x_{s_1}^{\alpha_1})^2) 
\end{equation*}
Again from $\smash{x_{s_1}^{\alpha_1} + x_{s_1}^{\beta} \doteq 1}$ and $\smash{x_{s_1}^{\beta} > 0}$
it follows $\smash{x_{s_1}^{\alpha_1} < 1}$ and thus
$\smash{(1 - (x_{s_1}^{\alpha _1})^2) > 0}$. This allows us to divide both sides by this term and we simply
get
\begin{equation*}
  (x_{u_1})_{s_1}^{\{\X\F\G a\}}  =0.5
\end{equation*}
Moreover, from the simplified constraints above, we have:
\begin{equation*}
(x_{u_5})_{s_1}^{\{\X\F\G a\}} 
  = x_{s_1}^{\alpha _1}\cdot (x_{u_1})_{s_1}^{\{\X\F\G a\}} + 0.5 \cdot x_{s_1}^{\beta} 
  = 0.5 \cdot (x_{s_1}^{\alpha _1} + x_{s_1}^{\beta})
  = 0.5 
\end{equation*}
From tree \circled{2} we get $\smash{(x_{u_2})_{s_1}^{\{\F\G a\}} = (x_{u_5})_{s_1}^{\{\X\F\G a\}}}$.
Substituting into $\smash{x_{s_1}^{\{\F\G a\}} \ge 0.3}$ from tree \circled{1} we get 
a tautology.

Altogether, the only non-trivial constraints in $\Gamma_\final$ are those
introduced by the various \IR{\textsf{A}} inferences for constraining probabilities
of actions. For the concretely $\CHOOSE$n alternatives in the example, this means
that any solution that satisfies  $\smash{x_{s_1}^{\alpha_1} > 0}$ and $\smash{x_{s_1}^{\beta} > 0}$
provides a policy. Notice that only $\smash{x_{s_1}^{\beta} > 0}$ is essential, and
$\CHOOSE$ing the alternative $\smash{x_{s_1}^{\alpha_1} \doteq 0}$, which entails
$\smash{x_{s_1}^{\beta} = 1}$, will do as well.

At this point it only remains to go through the derivations for tree \circled{4} and
\circled{5}. With the explanations so far this should be straightforward. 

Figure~\ref{fig:deriv-4} has the tree \circled{4} for the pivot $s_2 : \F \G a$. From
the MDP in Figure~\ref{fig:mdp-1} we expect it has probability one. It is formally
computed by adding the probabilities of $s_2 : \G a$, which is one, and of
$s_2 : \neg\G a, \X \F \G a$, which is zero. The latter requires identification of a
BSCC without yes-loops, as depicted, which, hence,  $\FORCE$es zero.
\begin{figure}[htp]
  \centering
  \includegraphics[scale=0.9]{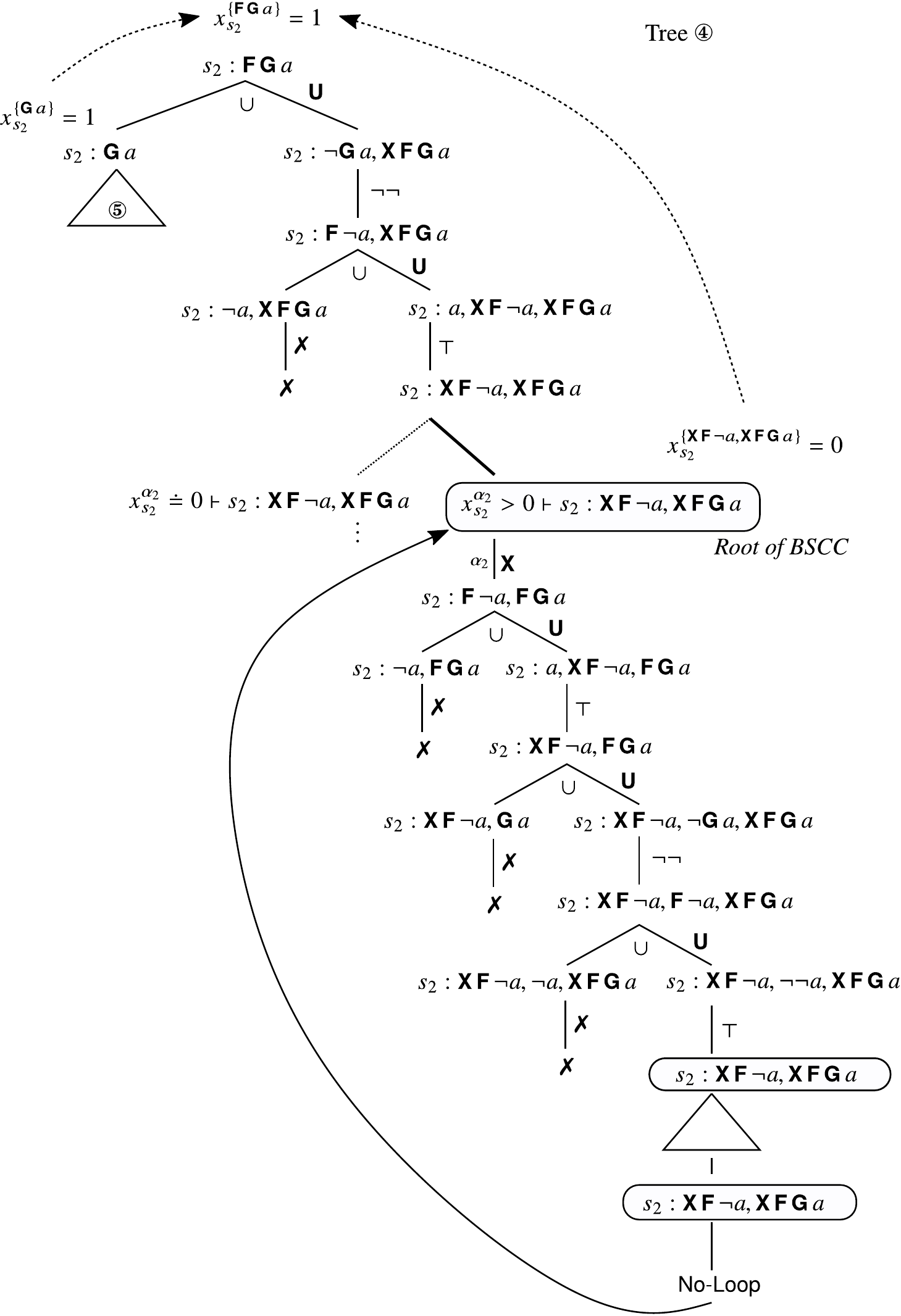}
  \caption{Sub-derivation with Tree \circled{4} continuing Figure~\ref{fig:deriv-3}.}
  \label{fig:deriv-4}
\end{figure}

Similarly, Figure~\ref{fig:deriv-5} has the tree \circled{5} for the pivot $s_1 s_2 : \G
a$. It is an example for a BSCC with a yes-loop. Notice that the poised pivot that forms
the BSCC has no $\TO X$-eventualities at all.
\begin{figure}[htp]
  \centering
  \includegraphics[scale=0.9]{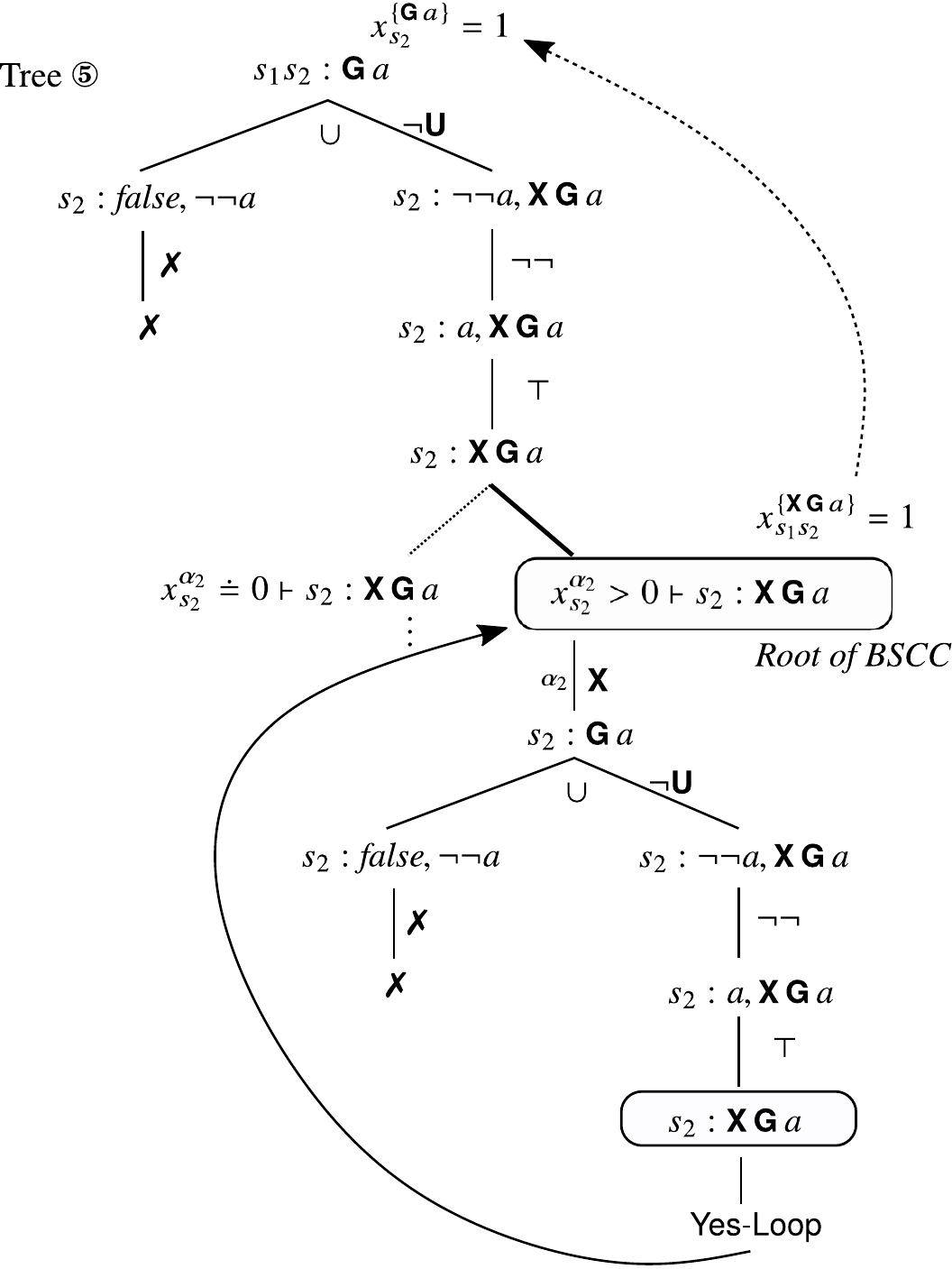}
  \caption{Sub-derivation with Tree \circled{5} continuing Figure~\ref{fig:deriv-4}.}
  \label{fig:deriv-5}
\end{figure}

\clearpage\newpage
\section{Proofs}
\label{sec:proofs}

\begin{lemma}
  \label{lemma:uvw}
  Let $b$ be a branch in a tree with poised nodes $u$ and $w$ with $u < w$.
  If $\Psi_u = \Psi_w$ then for every poised node $v$ with $u < v < w$ it holds 
  $\Psi_v \supseteq \Psi_u$.
\end{lemma}
\begin{proof}
  Suppose $\Psi_u = \Psi_w$ and a poised node $v$ with $u < v < w$.
  Let $x$ be the $\TO X$-successor node of $u$ on $b$. 
  The set $\Psi_u$ is of the form $\{ \X \psi_1, \ldots, \X \psi_n\}$ and, hence,
  $\Psi_x = \{ \psi_1, \ldots, \psi_n\}$.  

  Now consider the available inference rules. It is straightforward to check that
  every inference rule except $\TO U$ and $\neg \TO U$ either leaves the pivot $\pair m
  s :\Psi$
  untouched or replaces some 
  $\psi \in \Psi$ by zero or more strictly simpler formulas.
  With  ``$\psi_1$ is simpler than $\psi_2$'' we mean that $\psi_1$ has strictly less symbols
  than $\psi_2$ or else the number of symbols are the same but a negation sign in $\psi_1$ has been
  pushed inwards to get $\psi_2$ (this is needed for the $\neg \TO X$ rule). 

  The $\TO U$ and $\neg \TO U$ rules also add simpler subformulas to the conclusion,
  together with an $\TO X$-operator in front of the $\TO U$ formula or negated $\TO
  U$ formula it was applied to.

  The $\TO U$ and $\neg \TO U$ rules are the only ones that can restore the
  $\TO X$-operator applications in $\Psi_w = \{ \X \psi_1, \ldots, \X \psi_n\}$.
  In other words, the ancestor nodes of $w$ must collectively contain the formulas
  $\psi_1,\ldots,\psi_n$ in their formula sets and, furthermore, every $\psi_i$ is an $\TO U$-formula or a
  negated $\TO U$-formula.

  Because it is impossible to derive a $\TO U$-formula or a negated $\TO U$-formula
  from simpler formulas (there is just no inference rule capable of that) this means
  they must have been present from the beginning. No
  $\psi_i \in \Psi_x$ can have been replaced by simpler formulas from $x$ down to $w$.
  Furthermore,
  every $\psi_i$ must have been subject to a $\TO U$ and $\neg \TO U$ inference
  above the node $v$, otherwise $v$ is not poised. It follows $\Psi_v \supseteq \Psi_u$.
  \qed
\end{proof}

\begin{theorem*}[Soundness]{th:soundness}
  Let $\calM = (S, s_\init, A, P, L)$ be an MDP,
  $\pi_\fin = (M, \start, \Delta, \cdot)$ be
  a partially specified finite-memory policy with $\act$ unspecified,
  and $\phi$ a state formula. Suppose there is a program
  $\Gamma_\final := \GAMMA(\CHOOSE(\TABLEAU(
\{x_{\pair {\start(s_\init)} {s_\init}}^{\{\phi\}}\doteq 1\} \vdash \pair {\start(s_\init)} {s_\init}:\{\phi\})))$
  such that $\Gamma_\final$ is satisfiable.
  Let $\sigma$ be any solution of $\Gamma_\final$ and $\pi_\fin(\sigma)$ be the policy completed by $\sigma$.
  Then it holds $\calM, \pi_\fin(\sigma) \models \phi$.
\end{theorem*}

\begin{proof}
(Sketch)
  Let $\Gamma_\final$, $\sigma$ and $\pi_\fin(\sigma)$ as stated. We have to show
$\calM, \pi_\fin(\sigma) \models \phi$, equivalently
$\calM_{\pi_\fin(\sigma)}, \pair {\start(s_\init)} {s_\init} \models \phi$.

We need to look at the collection of sub-derivations from the initial sequent.
Let $\calT_0,\calT_1,\ldots,\calT_n$ be a sequence of trees such that
\begin{align*}
  \calT_0 & = \CHOOSE(\TABLEAU(\{x_{\pair {\start(s_\init)} {s_\init}}^{\{\phi\}}\doteq 1\} \vdash \pair {\start(s_\init)} {s_\init}:\{\phi\})) \\
  \calT_j & = \CHOOSE(\TABLEAU(\emptyset \vdash \pair {m_j} {s_j}: \{ \psi_j \}))\quad  \pbox[t]{ for $j=1\ldots n$, where $\calT_j$ 
            comes \\ from some \IR{$\TO P$} inference in the\\ derivation}
\end{align*}
For $i = 0\ldots n$ we call $\calT_i$ a \define{tree from $\Gamma_i \vdash \pair {m_i} {s_i} : \{\psi_i \}$}.

The sequence is meant to be minimal and closed under \IR{$\TO P$} inferences in 
subderivations. Formally, it is constructed inductively: starting from $\calT_0$ 
one takes the \IR{$\TO P$} inferences in $\calT_0$ and appends the 
trees for these \IR{$\TO P$} inferences. Then proceed to the next tree in the sequence, do the same,
and so on, until all sub-derivations have been processed.

We can accompany the sequence of trees with a sequence 
$\Gamma_0,\Gamma_1,\ldots,\Gamma_n$, where $\Gamma_i = \GAMMA(\calT_i)$, for $i = 0\ldots n$.
Together they represent a derivation from
$\{x_{\pair {\start(s_\init)} {s_\init}}^{\{\phi\}}\doteq 1\} \vdash \pair {\start(s_\init)} {s_\init}:\{\phi\}$ where $\Gamma_\final = \Gamma^0$.

To prove the theorem, we prove something more general.
For all $i = 0\ldots n$, if $\Gamma \vdash \pair {m} {s} : \Psi$ is the label of a node in $\calT_i$ then:
\begin{enumerate}
\item[(i)]   If $\Psi$ is a set of state
  formulas then $v \in \{0,1\}$ where $v = \sigma(x_{\pair {m} {s}}^{\Psi})$.
  Moreover, $v = 1$ iff $M_{\pi_\fin(\sigma)}, \pair m s \models \bigwedge \Psi$.
  (And hence $v = 0$ iff $M_{\pi_\fin(\sigma)}, \pair m s \not\models \bigwedge \Psi$.)
\item[(ii)] If $\Psi$ contains at least one proper path formula then
  \begin{equation*}
  \sigma(x_{\pair {m} {s}}^{\Psi}) = \mathop{\Pr^{\calM_{\pi_\fin(\sigma)}}}(\{r \in \textstyle{\Runs^{\calM_{\pi_\fin(\sigma)}}}(\pair {m} {s}) \mid \calM_{\pi_\fin(\sigma)}, r \models \bigwedge \Psi\})
\end{equation*}
\end{enumerate}
Notice that cases (i) and (ii) are exclusive and exhaustive. If $\Psi = \emptyset$ then case (i) applies.

For the proof, fix some $i \in \{ 0\ldots n\}$ and let $\Gamma \vdash \pair {m} {s} : \Psi$ be the label of a node in $\calT_i$.

\emph{Proof of (i).} 
Suppose $\Psi$ is a set of state formulas. We prove the conclusion by
induction on the structure of $\Psi$.

If $\Psi = \emptyset$ then with the $\Checkmark$ rule this is trivial. If there is a
classical $\phi \in\Psi$ then we have two cases: if $L(s) \models \phi$
(or equivalently: $\calM_{\pi_\fin(\sigma)},\pair m s \models \phi$) then we get a successor node with a
$\top$ inference. The result follows by induction together with the
constraint $\gammaone$ of the $\top$ rule.

If $L(s) \not\models \phi$ then with a $\Cross$ inference we get a constraint 
$x_{\pair {m} {s}}^{\Psi} \doteq 0$ in $\Gamma_\final$, which trivially gets the result. In particular it
follows $M_{\pi_\fin(\sigma)},\pair m s \not\models \bigwedge \Psi$.

Hence suppose now every $\phi \in \Psi$ is non-classical. Choose one such $\phi$.
If one of the $\neg\neg$, $\neg \TO P$, $\TO P \neg$, $\land$, $\neg\land$, $\PQ P1$, $\PQ P2$ or $\PQ P3$
inference rules is applicable, we get in each of the (one or two) conclusions a set of state
formulas. This is easy to verify by inspecting the inference rules. In each case the
result follows by induction together with the new constraints introduced by the
inference rules. We spell this out only in the most interesting case, the $\neg\land$ rule.

For the $\neg\land$ rule, the formula set $\Psi$ is of the
form $\{ \neg(\psi_1 \land \psi_2) \} \cup \Psi'$ and $\phi = \neg(\psi_1 \land \psi_2)$. 
For the two conclusions we get results by induction.
Notice that the left conclusion has $\neg \psi_1$ in its formula set $\{\neg\psi_1\} \cup \Psi'$, whereas the
right conclusion has the complement $\psi_1$ in its formula set $\{\psi_1, \neg \psi_2\} \cup \Psi'$.
This entails for the induction results that not both $M_{\pi_\fin(\sigma)},\pair m s \models \{\neg\psi_1\} \cup \Psi'$ and
$M_{\pi_\fin(\sigma)}, \pair m s  \models \{\psi_1, \neg \psi_2\} \cup \Psi'$ can hold. In terms of corresponding variables this
means $v_l, v_r \in \{0,1\}$  but $\{v_l, v_r\} \neq \{1\}$ where 
$v_l = \sigma(x_{\pair {m} {s}}^{\{\neg\psi_1\} \cup \Psi'})$ and
$v_r = \sigma(x_{\pair {m} {s}}^{\{\psi_1, \neg \psi_2\} \cup  \Psi'})$.
Finally, with the definition of the constraint $\gamma$ in the definition
of the $\neg\land$ rule the result follows easily.

The only missing case is when $\phi$ is of the form $\P_{\sim z} \psi$. By closure under the
inference rules with higher priority (those mentioned above) we know that $\psi$ is a
proper path formula.

If the $\PQ P$ rule is applicable then the left conclusion adds the constraint
$\gammaleft = x_{\pair {m} {s}}^{\{\psi\}} \sim z$ and the right conclusion adds the constraint
$\gammaright = x_{\pair {m} {s}}^{\{\psi\}} \eqsim z$.
As the $\PQ P$ is a don't know rule the tree $\calT_k$ has $\CHOOSE$n one of them.

In any case, the \IR{$\PQ P$} rule invokes a tableau derivation which gives us some
tree $\calT_m$ from $\emptyset \vdash {\pair {m} {s}}: \{\psi\}$ in the above sequence of trees. From property (ii)
we get
\begin{equation*}
  \sigma(x_{\pair {m} {s}}^{\Psi}) = \mathop{\Pr^{\calM_{\pi_\fin(\sigma)}}}(\{r \in \textstyle{\Runs^{\calM_{\pi_\fin(\sigma)}}}(\pair {m} {s}) \mid \calM_{\pi_\fin(\sigma)}, r \models \psi\})
\end{equation*}

In the left case, with $(x_{\pair {m} {s}}^{\{\psi\}} \sim z) \in \Gamma_\final$ and the semantics of the
$\PQ P$ operator it follows $\calM_{\pi_\fin(\sigma)},\pair m s \models \P_{\sim z} \psi$.  (Recall that $\Gamma_\final$ is
satisfiable.) With a, not necessarily
immediately, following \IR{$\PQ P \top$} inference the situation is the same as above,
where we had the classical formula $\phi$ such that
$\calM_{\pi_\fin(\sigma)}, \pair m s \models \phi$ and a $\top$ inference. As above, the result follows by induction.

In the right case, with $(x_{\pair {m} {s}}^{\{\psi\}} \eqsim z) \in \Gamma_\final$ and
the semantics of the $\PQ P$ operator it follows $\calM_{\pi_\fin(\sigma)}, \pair m s \not\models \P_{\sim z} \psi$.
With a (not necessarily immediately) following \IR{$\PQ P \Cross$} the result
follows trivially.

If the $\PQ P$ rule is not applicable then $\Gamma$ contains $\gammaleft$ or
$\gammaright$ already, but the same argumentation as for the left/right case applies.

This concludes the proof of property (i).

\emph{Proof of (ii).} The proof is similar to the proof of (i) except for BSCCs,
which requires special consideration.

Let $u$ be the node in $\calT_i$ labelled with $\Gamma \vdash {\pair {m} {s}} : \Psi$, the sequent we are looking at.

What we do is induction on the structure of $\Psi$, much like in case (i).
We are adding up probabilities from the children of
a union-branching to the parent node, in terms of corresponding variables.
That this is correct is clear from inspecting the design of the
inference rules with respect to the PCTL* semantics.
Verifying the $\TO X$ rule may require a little closer look, though, but is not too
difficult. An important detail is that the branching out into $\TO X$-successor nodes
happens according to prescribed actions, by preference of the \textsf{A} rule over
the $\TO X$ rule. If, say, $\alpha \in \Prescribed(\Gamma,{\pair {m} {s}})$ is a prescribed action, then $\Gamma$
(and hence $\Gamma_\final$) contains a variable $x_{\pair {m} {s}}^\alpha$. It follows that the policy
$\pi_\fin(\sigma)$ will accordingly have $\act_\sigma(m, s,\alpha) = \sigma(x_{\pair {m} {s}}^\alpha)$  rather than an arbitrary
choice for making it complete.


If $u$ is also the root of a BSCC in $\calT_i$ then $\FORCE$ing adds $x_{\pair {m} {s}}^\Psi \doteq 1$ or
$x_{\pair {m} {s}}^\Psi \doteq 0$ to $\GAMMA(\calT_i)$ and hence also to $\Gamma_\final$.
This entails $\sigma(x_{\pair {m} {s}}^\Psi) = \chi$ where $\chi$ or $\chi = 0$, respectively.
For proving (ii) this means we need to show
\begin{equation}
  \label{eq:1}
  \mathop{\Pr^{\calM_{\pi_\fin(\sigma)}}}(\{r \in \textstyle{\Runs^{\calM_{\pi_\fin(\sigma)}}}(\pair
  {m} {s}) \mid \calM_{\pi_\fin(\sigma)}, r \models \bigwedge \Psi\}) = \chi
\end{equation}

Let us start by some considerations about runs in relation to the BSCC rooted at $u$.
Consider any run $r$ from
$\pair {m} {s}$. We can trace $r$'s state transitions in $\Subtree(u)$ as follows: at a
poised inner node (such as $u$) select the proper $\TO X$-successor node, say, $v$ as
given by the state transition in $r$ we are currently looking at (initially the first
one, from $u$).  Now consider
$\Subtree(v)$. It cannot be a 0-deadend by definition of ambiguity. In other words,
some of its branches may lead to a $\Cross$-ed leaf without encountering a poised
inner node, but not all of them (and no branch has a $\Checkmark$-ed leaf
either). For any branch of the latter kind we can find the first poised node as we go
down. Let $v_1,\ldots,v_n$ be all these poised \define{fringe} nodes, for some
$n \ge 1$. We can view $\Subtree(v)$ as an analysis of the Boolean structure until
reaching a fringe node. The path component does not advance during that, and so the
fringe nodes each have the same path component as $v$. This is why we can take
\emph{any} $v_i$ for continuing the tracing (with $p_{v_i} = p_v$). Should $v_i$ be a
leaf node we need to follow its backlink to the inner, indistinguishable node first.

This way we can follow $r$ stepwise in the BSCC rooted at $u$.

Let us consider the case $\chi = 1$ in (\ref{eq:1}) above. That is, we have a
\textsf{Yes-Loop}. 

Should we have the freedom to construct $r$ as we wish, we could do it in such a way
that it always passes through the (or one of the) \textsf{Yes-Loop} leaves. This
would give us what is sometimes called a ``lasso'': from $u$ go down the proper branch
and find that inner node, say $w$ ($u = w$ is possible), that the \textsf{Yes-Loop}
leaf backjumps to. Then there is an initial segment from $u$ to $w$ followed by
circles leading back to $u$. The run $r$ would just follows this ``path'' ad infinitum.

This situation is identified by Mark Reynolds for proving his tableau
algorithm~\cite{Reynolds:LTL-Tableaux:GandALF:2016} correct. (The
long version~\cite{Reynolds:LTLTableau:long:2016} has proofs.)
Let us call his tableau ``MR tableaux''. An important difference is
that~\cite{Reynolds:LTLTableau:long:2016} is concerned with LTL
satisfiability, not model checking. (A minor difference is that we have additionally
the $\PQ P$ operator, which however behaves much like a propositional symbol by the
recursive call to the algorithm, plus induction.)
For the soundness proof he assumes as given an
LTL formula and an MR tableau that has (in our words) a \textsf{Yes-Loop} leaf.
From the \textsf{Yes-Loop} he constructs a lasso run as described above,
derives states and a labelling function from it, and shows that this run satisfies
the given LTL formula. He actually has a more general construction which annotates
the states in the runs with poised formula sets and keeps track of the expansions by
the $\TO X$ rule and all other subsequent inferences (which are similar to ours).
See the ``truth lemma'', Lemma 3 in~\cite{Reynolds:LTLTableau:long:2016}. 

We wish to re-use the correctness results for the MR tableaux. The main difficulty is
that, of course, we are not free to construct the run, instead it is given. More
precisely, the given run $r$ might not only follow the lasso, it may include other
segments leading to \textsf{No-Loop}s as well. What we do know, however, is that for
fairness reasons every leaf must be visited infinitely often (recall we are dealing
with a BSCC). This means that the loop in the lasso must be
executed infinitely often along $r$, but not necessarily consecutively.

Because of the just said the truth lemma is not immediately applicable to our run
$r$. But we can think of the run $r$ as executing a modified \textsf{Yes-Loop} loop infinitely often,
with execution segments detouring into \textsf{No-Loop}s spliced in.
By inspecting the proof of the truth lemma it becomes clear that this splicing-in
does not hurt as long as some essential formulas are preserved between
subsequent poised nodes. These formulas are used to establish that $\TO U$-formulas
or negated $\TO U$-formulas, stemming from an $\TO X$ inference to a poised node,
are satisfied by the run in the MR tableau. See Lemmas 2 and 3
in~\cite{Reynolds:LTLTableau:long:2016}. That these formulas are indeed
preserved by our runs as well then follows from Lemma~\ref{lemma:uvw} above: if
there are additional poised nodes spliced into the \textsf{Yes-Loop}s, they are equal to
or supersets of the original poised nodes and, therefore, harmless.
In other words, it can be shown that our run $r$ satisfies $\Psi$.

Finally as our run $r$ is chosen arbitrarily, every run from $\last(p)$ satisfies
$\Psi$. As the probability of the set of all runs is one, we have shown (\ref{eq:1}) for the case $\chi=1$.

It remains to consider the case $\chi = 0$. 
That is, we have no \textsf{Yes-Loop} in the BSCC rooted at $u$. 
Again we refer to the proofs in~\cite{Reynolds:LTLTableau:long:2016}. This time
it is actually easier. The completeness proof
in~\cite{Reynolds:LTLTableau:long:2016} (Lemma 5) needs
to analyze an arbitrary given run and trace it down the branches in the MR
tableau, much like our run $r$ is allowed to walk down any branch and return to an
inner node from its leaves (which are all \textsf{No-Loop}). So this situation is a
better match.

The main argument, in brief, is that there must be an $\TO X$-eventuality in $\Psi$ that
remains unsatisfied along $r$, making $\Psi$ unsatisfied by $r$, if there is no
\textsf{Yes-Loop}. For this it is shown that the loop-check does not prematurely cut
branches. In brief, when a leaf is a \textsf{No-Loop} the branch leading to it
repeatedly failed to satisfy an $\TO X$-eventuality and it would not help to make this
same mistake twice.

Finally as our run $r$ is chosen arbitrarily, no run from $\pair {m} {s}$ satisfies
$\Psi$. As the probability of the empty set of runs is zero, we have shown (\ref{eq:1}) for the case $\chi=0$.

This concludes the proof of property (ii).

Finally it is easy to prove the conclusion $\calM, \pi_\fin(\sigma) \models \phi$ of the
theorem, which expands to the equivalent statement
$\calM_{\pi_\fin(\sigma)}, \pair {\start(s_\init)} {s_\init} \models \phi$. 
Because $\phi$ is a state formula we can apply property (i) to
the tree $\calT_0$ from $\{x_{\pair {\start(s_\init)} {s_\init}}^{\{\phi\}}\doteq 1\} \vdash {\pair {\start(s_\init)} {s_\init}}:\{\phi\}$ and conclude
$v \in \{0, 1\}$ where $v = \sigma(x_{\pair {\start(s_\init)} {s_\init}}^{\{\phi\}})$. With
$x_{\pair {\start(s_\init)} {s_\init}}^{\{\phi\}}\doteq 1 \in \Gamma_\final$ trivially $v = 1$. By property (i) again 
$\calM_{\pi_\fin(\sigma)}, \pair {\start(s_\init)} {s_\init} \models \phi$.\qed
\end{proof}

\begin{theorem*}[Completeness]{th:completeness}
  Let $\calM = (S, s_\init, A, P, L)$ be an MDP,
  $\pi_\fin = (M, \start, \Delta, \act)$ 
  a finite-memory policy, and $\phi$ a state formula. Suppose
  $\calM, \pi_\fin \models \phi$. Then there is a satisfiable program
  $\Gamma_\final := \GAMMA(\CHOOSE(\TABLEAU(\{x_{\pair {\start(s_\init)}
    {s_\init}}^{\{\phi\}}\doteq 1\} \vdash \pair {\start(s_\init)} {s_\init}:\{\phi\})))$
  and a solution $\sigma$  of $\Gamma_\final$ such that $\act_\sigma(m,s,\alpha) = \act(m,s,\alpha)$
  for every pair $\pair m s$ in the policy domain of $\Gamma_\final$.
  Moreover $\calM, \pi_\fin(\sigma) \models \phi$.
\end{theorem*}

\begin{proof}
(Sketch.)
Let $\phi$ be as stated and suppose there is a finite-memory policy $\pi_\fin$ such that
$\calM, \pi_\fin \models \phi$. In other words $\calM_{\pi_\fin}, \pair  {\start(s_\init)}
{s_\init} \models \phi$.

A preliminary: a \define{binding} is a pair$(x,v)$, usually written as $x \mapsto v$, where
$x$ is a variable and $v \in [0,1]$. A \define{substitution} is a finite set of
bindings.  The set $\dom(\sigma) = \{ x \mid (x \mapsto v) \in \sigma \text{ for some $v$ }\}$ is
called the \define{domain of $\sigma$}. We use substitutions to construct solutions of
programs. 

The proof plan is to construct $\TABLEAU(\{x_{\pair {\start(s_\init)} {s_\init}}^{\{\phi\}}\doteq 1\} \vdash \pair {\start(s_\init)} {s_\init}:\{\phi\})$ in a similar
fashion as one would do in the soundness proof for classical tableau, by analysing $\phi$
syntactically and going down branches. It is a bit trickier in the completeness case, though.
We have to $\CHOOSE$ along the way as needed for complying with the prescribed
actions in the given policy $\pi_\fin$ and the $\PQ P$-formulas. We also have to show that the
$\GAMMA$ operations results in a satisfiable $\Gamma_\final$. The rest of the proof spells
this out in more detail.

We need something general to keep the induction going, as follows.

As in the soundness proof we work with a sequence of trees
$\calT_0,\calT_1,\ldots,\calT_n$ such that
\begin{align*}
  \calT_0 & = \CHOOSE(\TABLEAU(\{x_{\pair {\start(s_\init)} {s_\init}}^{\{\phi\}}\doteq 1\} \vdash \pair {\start(s_\init)} {s_\init}:\{\phi\})) \\
  \calT_j & = \CHOOSE(\TABLEAU(\emptyset \vdash \pair {m_j} {s_j}: \{ \psi_j \}))\quad  \pbox[t]{ for $j=1\ldots n$, where $\calT_j$ 
            comes \\ from some \IR{$\TO P$} inference in the\\ derivation}
\end{align*}
Again, for $i = 0\ldots n$ we call $\calT_i$ a \define{tree from $\Gamma_i \vdash \pair
{m_i} {s_i} : \{\psi_i \}$}, and we accompany the sequence of trees with a sequence 
$\Gamma_0,\Gamma_1,\ldots,\Gamma_n$, where $\Gamma_i = \GAMMA(\calT_i)$, for $i = 0\ldots n$.
Together they represent a derivation from
$\{x_{\pair {\start(s_\init)} {s_\init}}^{\{\phi\}}\doteq 1\} \vdash \pair {\start(s_\init)} {s_\init}:\{\phi\}$ where $\Gamma_\final = \Gamma^0$.

Unlike as in the soundness proof, these sequences are not given a priori. Indeed, we have
to show they exist. We do this iteratively with the help of
a couple of variables, collectively called the \define{induction variables}:
\begin{itemize}
\item The sequence of the $\calT_i$'s, initialized with a one-tree sequence
$\calT_0$  with a root node only labelled with 
$\{x_{\pair {\start(s_\init)} {s_\init}}^{\{\phi\}}\doteq 1\} \vdash \pair {\start(s_\init)}{s_\init}:\{\phi\}$
\item The sequence of the $\Gamma_i$'s initialized with a one-program sequence\\
  $\Gamma_0=\{x_{\pair {\start(s_\init)} {s_\init}}^{\{\phi\}}\doteq 1\}$.
\item A current substitution $\sigma$, initialized with
  \begin{equation*}
    \sigma = \{ x_{\pair m s}^\alpha \mapsto \act(m,s,\alpha) \mid \pair m s \in M \times S \text{ and } \alpha \in A(s) \} 
    \cup \{ x_{\pair {\start(s_\init)} {s_\init}}^{\{\phi\}} \mapsto 1 \}
  \end{equation*}
\end{itemize}
To extend this initial state to a derivation we pick any leaf in any $\calT_i$ and
apply any inference rule to it, subject only to preference constraints. (This freedom
is needed to match the claim that the inference rules can be applied in a don't-care
fashion, subject only to the preference constraints.)

On applying an inference rule we will show that the following \define{invariant}
is preserved: if $\Gamma \vdash \pair {m} {s}:\Psi$ is the label of a node in some $\calT_i$ then all of the
following holds:
\begin{itemize}
\item[(i)] $x_{\pair {m} {s}}^{\psi} \in \dom(\sigma)$ and $\sigma$ is a solution of $\Gamma_j$, for all $j = 0\ldots n$.
\item[(ii)] $\act_\sigma(m, s ,\alpha) = \act(m, s,\alpha)$ for every $\pair {m} {s} \in
M\times S$ such that $\pair {m} {s}$ is in the  policy-domain of $\Gamma_0$.
\item[(iii)]
  if $\Psi$ is a set of state formulas then
      \begin{quote}
        $\calM_{\pi_\fin(\sigma)}, \pair m s \models \bigwedge \Psi$\quad iff\quad
      $\calM_{\pi_\fin}, \pair m s \models \bigwedge \Psi$ \quad iff\quad
      $\sigma(x_{\pair {m} {s}}^\Psi) = 1$.

      Moreover $\sigma(x_{\pair {m} {s}}^\Psi) \in \{0,1\}$.
      \end{quote}
    \item[(iv)]
      if $\Psi$ contains at least one proper path formula then
      \begin{align*}
        \sigma(x_{\pair {m} {s}}^\Psi)  & = \textstyle\Pr^{\calM_{\pi_\fin(\sigma)}}(\{r \in \Runs^{\calM_{\pi_\fin(\sigma)}}(\pair {m} {s}) \mid \calM_{\pi_\fin(\sigma)},r \models \bigwedge \Psi\})\\
&= \textstyle\Pr^{\calM_{\pi_\fin}}(\{r \in \Runs^{\calM_{\pi_\fin}}(\pair {m} {s}) \mid
        \calM_{\pi_\fin},r \models \bigwedge \Psi\})
      \end{align*}
    \end{itemize}
We refer to (i)--(iv) collectively as the \define{invariant}.

The invariant holds initially. This is trivial for each (i)--(iv).
We will show below that the invariant is preserved by applying an inference rule,
anywhere. This requires updating the induction variables appropriately. This process
will end in the announced derivation, and, by construction, $\Gamma_0$ will be $\Gamma_\final$.

Once this is done, the theorem is proved easily:
\begin{itemize}
\item That $\Gamma_\final$ is satisfiable is trivial with (i).
  \item That $\act_\sigma(m,s,\alpha) = \act(m,s,\alpha)$ for every $m\times s$ in the
    policy-domain of $\Gamma_\final$ becomes identical to (ii).
    \item $\calM_{\pi_\fin(\sigma)}, \pair {\start(s_\init)} {s_\init} \models \phi$ follows
      from (iii) and $\calM_{\pi_\fin}, \pair {\start(s_\init)} s_\init \models \phi$
      as given.
\end{itemize}

Hence it only remains to prove the invariant.

Choose some $\calT_i$ and some leaf in $\calT_i$ arbitrarily. Let $\Gamma \vdash \pair {m} {s}:\Psi$ be
the leaf's label. Assume the invariant holds for $\Gamma \vdash \pair {m} {s}:\Psi$.
Choose any inference rule applicable to $\Psi$, not violating preference constraints,
and apply it. Then show that the invariant holds afterwards.

In the first case $\Psi$ is a set of state formulas.
Partition $\Psi = \{\phi\} \uplus \Psi'$ to match the form of the inference rule. 

If the rule is  $\top$,  add $\gammaone\ =\ x_{\pair {m} {s}}^{\{\phi\} \uplus \Psi'} \doteq x_{\pair {m} {s}}^{\Psi'}$
to $\Gamma_i$. All inference rules are mutually exclusive for a fixed selected formula
$\phi$. Thus, if $\Gamma_i$ , or any other $\Gamma_j$ for $i \neq j$, already contains an equality for 
$x_{\pair {m} {s}}^{\{\phi\} \uplus \Psi'}$ it must be the same $\gammaone$. In this case there is nothing
  to show for (i) and (ii) to carry over.
  If, otherwise $\gammaone$ is fresh then we extend $\sigma$ by the new
  binding $x_{\pair {m} {s}}^{\{\phi\} \uplus \Psi'} \mapsto \sigma(x_{\pair {m} {s}}^{\Psi'})$.
  Because $x_{p}^{\{\phi\} \uplus \Psi'}$ is fresh the new binding cannot interfere with those
  already in the domain of $\sigma$ before the binding was added.
  In any case 
  \begin{equation}
    \label{eq:p:1}
    \sigma(x_{\pair {m} {s}}^{\{\phi\} \uplus \Psi'}) = \sigma(x_{\pair {m} {s}}^{\Psi'})\enspace.
  \end{equation}
  
  The applicability condition of $\top$ tells us
  $L(s) \models \phi$. With the semantics of PCTL* it follows
  \begin{align}
\textstyle  \calM_{\pi_\fin(\sigma)}, \pair m s \models \phi \land \bigwedge \Psi' & \text{ iff }
\textstyle  \calM_{\pi_\fin(\sigma)},  \pair m s \models \bigwedge \Psi' \label{eq:p:3} \\
\textstyle    \calM_{\pi_\fin}, \pair m s \models \phi \land \bigwedge \Psi' & \text{ iff }
\textstyle    \calM_{\pi_\fin}, \pair m s \models \bigwedge \Psi' \label{eq:p:2} 
  \end{align}
  Substituting (\ref{eq:p:1}), (\ref{eq:p:3}) and (\ref{eq:p:2}) into the iff-chain in (iii) gives
      \begin{quote}
      $\calM_{\pi_\fin}, \pair m s \models \bigwedge \Psi'$\quad iff\quad
      $\calM_{\pi_\fin(\sigma)}, \pair m s \models \bigwedge \Psi'$ \quad iff\quad
      $\sigma(x_{\pair {m} {s}}^{\Psi'}) = 1$.
    \end{quote}
    which proves (iii) for the new leaf.
    For (iv) there is nothing to show. 
    This concludes the proof for the case that the $\top$ rule is applied.

    The proofs for the inference rules
    $\Cross$, $\Checkmark$, $\neg\neg$, $\neg\PQ P$, $\PQ  P\neg$, $\land$, $\neg\land$, $\PQ P1$, $\PQ P2$,
    $\PQ P3$ all are analogous to the $\top$ rule. We do not carry them out in detail.
    Notice again (as in the soundness proof) that in the $\neg\land$ rule the
    left and right conclusions are mutually exclusive, which enables correctly taking
    the sum of the corresponding variables. 

    It remains to consider the inference rules for $\PQ P$-formulas. The rules are
    $\PQ Q$, $\PQ P\top$ and $\PQ P\Cross$.

    If $\PQ P$ is applicable then there is a formula
    $\P_{\sim z} \phi \in \Psi$ such that $\gammaleft \notin \Gamma_i$ and $\gammaright
    \notin \Gamma_i$.

    The $\PQ P$ inference starts a new derivation from $\emptyset \vdash \pair {m} {s}: \{\phi\}$.
    In terms of our proof it starts with a tree $\calT_k$ for $\emptyset \vdash \pair {m} {s}: \{\phi\}$, for
    some $k >i$. The $\Gamma'$ mentioned in the $\PQ P$-rule is 
    $\Gamma_k$ after that derivation has finished.
    Hence assume that this derivation (and possibly further sub-derivations)
    have already been carried out. 
    We get as a result new trees
    starting from $\calT_k$ and new programs starting from $\Gamma_k$ and the invariant
    will be preserved by induction. 
    Property (i) then tells us that $\sigma$ is a solution of $\Gamma_j$, for all $j = 0\ldots n$.
    Now, with $\Gamma \subseteq \Gamma_i$ (by construction, we add all $\Gamma$'s in
    the pivots as we encounter them) and the fact that in the $\PQ P$ inference the
    program $\Gamma'$ is just $\Gamma_k$ it follows that $\sigma$ is a solution for 
    $\Gamma \cup \Gamma'$ and we could add $\Gamma \cup \Gamma'$ to $\Gamma_i$ without affecting 
    (i).

    The conclusion of the $\PQ P$-inference, however, is either
    $\Gamma \cup \Gamma ' \cup \{\gammaleft\}$ or $\Gamma \cup \Gamma' \{\gammaright\}$.
    Notice we already have the tree $\calT_k$ for $\emptyset \vdash \pair {m} {s} : \{\phi\}$. By
    invariant (iv) $\sigma$ will contain a binding for $x_{\pair {m} {s}}^{\{\phi\}}$
    (which is $x_{\pair {m} {s}}^{\{\phi\}} \mapsto \Pr^{\calM_{\pi_\fin}}(\{r \in \Runs^{\calM_{\pi_\fin}}(\pair {m} {s}) \mid \calM_{\pi_\fin},r \models \phi\})$).
    This binding satisfies $\gammaleft$ or $\gammaright$ and, hence, prescribes
    $\CHOOSE$ing the left or the right conclusion for extending $\calT_i$ so that 
    invariant (i) is preserved. 
    Properties (ii), (iii) and (iv) are all trivially preserved by the extension. 
    
    The open cases are the $\PQ \top$and the $\PQ P\Cross$ rules. Their proofs are
    analogous to the $\top$and the $\Cross$ rules and are omitted.

    This completes the proof for the case that $\Psi$ is a set of state formulas.  Hence
    assume now that $\Psi$ contains at least one state formula.  The relevant invariant
    properties are (i), (ii) and (iv).

    The proofs are analogous to the ones when $\Psi$ is a set of state formulas.
    Essentially, the invariant follows from the design of the inference rules.
    This holds true also for the $\TO X$-rule, which is a bit tedious to inspect.
    One has to verify that it correctly reflects the transition
    probability function of the Markov chain $\calM_{\pi_\fin}$, cf.\
    Section~\ref{sec:preliminaries}.
    Two important issues, though.

    First, the \textsf{A} inferences are preceding the $\TO X$-inferences.
    Unlike as in the soundness proof, the \textsf{A} inferences must be carried out
    consistently with the given policy $\pi_\fin$. More precisely, we $\CHOOSE$ the left conclusion
    $x_{\pair {m} {s}}^\alpha \doteq 0$ in an \textsf{A} inference 
    if $\act(m, s, \alpha) = 0$, and we $\CHOOSE$ the right conclusion
    $x_{\pair {m} {s}}^\alpha > 0$ otherwise. The variable $x_{\pair {m} {s}}^\alpha$ is already in the domain of
    $\sigma$ because a binding to the value $x_{\pair {m} {s}}^\alpha \mapsto \act(m, s, \alpha)$ was put there initially.
    This shows that the invariant is preserved for \textsf{A} inferences.

    At the end of the derivation all bindings for the variables $x_{\pair {m} {s}}^\alpha$
    such that $\pair {m} {s}$ is not in the policy domain of any $\calT_j$ can be
    removed from $\sigma$.
    This allows us to prove the missing detail of invariant (ii).
    
    Second, BSCCs and $\FORCE$ing. Suppose the node we are looking at is $u$, the one
    with the sequent $\Gamma \vdash \pair {m} {s}:\Psi$. 
    The derivation adds a constraint $x_{\pair {m} {s}}^\psi \doteq \chi$ for some $\chi \in  \{0,1\}$ if $u$ is
    the root of a BSCC in $\calT_i$. 
    We need to argue for the correctness of doing that.

    That $u$ is the root of a BSCC has a meaning in terms of the Markov Chain induced
    by the MDP $\calM$ and the policy $\pi_\fin$: the policy state $\pair {m} {s}$ is an entry node into a BSCC in
    the state  transition diagram of this Markov Chain. The state transitions are given
    by the actions prescribed by $\pi_\fin$, and the same actions are prescribed for the
    state transitions in the BSCC rooted at $u$. This follows from the
    \textsf{A} inferences as just explained. 
    This is why the runs from $\pair {m} {s}$ in the Markov Chain are exactly the same as
    the runs  from $\pair {m} {s}$ starting from $u$ in $\calT_i$ by going down branches
    and following backlinks. Because in BSCCs in Markov Chains
    either all (fair) runs or no run satisfy $\Psi$, the probability $\Pr^{\calM_{\pi_\fin}}(\{ r \in
    \Runs^{\calM_{\pi_\fin}}(\pair {m} {s}) \mid \calM_{\pi_\fin},r  \models \bigwedge \Psi\})$ is either 1 or 0.
    By the correctness of $\FORCE$ing (cf.\ the
    soundness proof)
    the corresponding constraint $x_{\pair {m} {s}}^{\Psi} \doteq 1$ or $x_{\pair {m} {s}}^{\Psi} \doteq 0$ will be correctly
    added in $\Gamma_i$. This is sufficient to prove the invariant (iv).\qed
  \end{proof}


\end{document}